%% file: EPR-quantum.tex
\documentclass[a4paper,onecolumn,11pt,accepted=2023-01-31]{quantumarticle}
\pdfoutput=1

\usepackage{revsymb, amsmath, amsfonts, amssymb, enumerate, fullpage, amsthm, graphicx, braket, relsize, bbm, mathrsfs, mathtools}
\usepackage[normalem]{ulem}

\usepackage{graphicx,color}
\usepackage{paralist}
\usepackage{subcaption}
\usepackage[labelformat=simple]{subcaption}

\usepackage[numbers,sort&compress]{natbib}
\usepackage{xfrac}

\usepackage[linktocpage=true, colorlinks=true, linkcolor=red, urlcolor=blue, citecolor=blue]{hyperref}

\newtheorem{theo}{Theorem}
\newtheorem{thm}[theo]{Theorem}
\newtheorem{prop}[theo]{Proposition} 
\newtheorem{lem}[theo]{Lemma}
\newtheorem{cor}[theo]{Corollary}

\newtheorem{defn}[theo]{Definition}
\newtheorem{rem}[theo]{Remark}

\newtheorem{sdp}[theo]{SDP}

\newcommand{\etal}{{\it{et al.}}}

\newcommand{\cH}{\mathcal{H}}

\newcommand{\id}{\mathbb{I}}
\newcommand{\tr}[2]{\mathrm{tr}_{#2} \left\{ #1 \right\}}
\newcommand{\Tr}[1]{\mathrm{tr}\left\{#1 \right\}}

\newcommand{\av}[1]{\langle #1 \rangle}

\newcommand{\X}{\mathbb{X}}

\newcommand{\A}{\mathbb{A}}
\newcommand{\As}{\boldsymbol{\Sigma}}

\newcommand{\FA}{\mathsf{S}}
\newcommand{\FM}{\mathsf{M}}
\newcommand{\M}{\mathrm{M}}
\newcommand{\I}{\mathsf{I}}

\begin{document}

\title{Quantifying EPR: the resource theory of  nonclassicality of common-cause assemblages}

\author{Beata Zjawin}
\affiliation{International Centre for Theory of Quantum Technologies, University of  Gda{\'n}sk, 80-308 Gda{\'n}sk, Poland}
\email{beata.zjawin@phdstud.edu.ug.pl}
\author{David Schmid}
\affiliation{International Centre for Theory of Quantum Technologies, University of  Gda{\'n}sk, 80-308 Gda{\'n}sk, Poland}
\affiliation{Perimeter Institute for Theoretical Physics, 31 Caroline St. N, Waterloo, Ontario, N2L 2Y5, Canada}
\affiliation{Institute for Quantum Computing and Department of Physics and Astronomy, University of Waterloo, Waterloo, Ontario N2L 3G1, Canada}
\author{Matty J.~Hoban}
\affiliation{Cambridge Quantum Computing Ltd, 9a Bridge Street, Cambridge, CB2 1UB, United Kingdom}
\affiliation{Department of Computing, Goldsmiths, University of London, New Cross, London SE14 6NW, United Kingdom}
\author{Ana Bel\'en Sainz}
\affiliation{International Centre for Theory of Quantum Technologies, University of  Gda{\'n}sk, 80-308 Gda{\'n}sk, Poland}

\maketitle

\begin{abstract}
 Einstein-Podolsky-Rosen (EPR) steering is often (implicitly or explicitly) taken to be evidence for spooky action-at-a-distance. 
An alternative perspective on steering is that Alice has no causal influence on the physical state of Bob's system; rather, Alice merely updates her knowledge of the state of Bob's system by performing a measurement on a system correlated with his. 
In this work, we elaborate on this perspective (from which the very term `steering' is seen to be inappropriate), and we are led to a resource-theoretic treatment of correlations in EPR scenarios. For both bipartite and multipartite scenarios, we develop the resulting resource theory, wherein the free operations are local operations and shared randomness (LOSR).
 We show that resource conversion under free operations in this paradigm can be evaluated with a single instance of a semidefinite program, making the problem numerically tractable. Moreover, we find that the structure of the pre-order of resources features interesting properties, such as infinite families of incomparable resources. In showing this, we derive new EPR resource monotones. 
We also discuss advantages of our approach over a pre-existing proposal for a resource theory of `steering', and discuss how our approach sheds light on basic questions, such as which multipartite assemblages are classically explainable.
\end{abstract}

\section{Introduction}

The notion of Einstein-Podolsky-Rosen (EPR) `steering'~\cite{einstein1935can,schrodinger1935discussion,cavalcanti2009experimental} refers to a form of nonclassical correlations that arise when one considers measurements performed on half of a bipartite system prepared on an entangled state. Such correlations have multiple applications in quantum information~\cite{wiseman2007steering,uola2020quantum}; for example, they allow for the certification of entanglement under relaxed assumptions~\cite{cavalcanti2015detection,mattar2017experimental}, and constitute an information-theoretic resource for various cryptographic tasks~\cite{branciard2012one,gianissecretsharing}. For these reasons, many previous works~\cite{cavalcanti2009experimental,pusey2013negativity,weight,robustness,gallego2015resource} have begun the process of characterizing the resourcefulness of such correlations.

In this paper, we develop a novel  resource-theoretic approach to quantifying the nonclassicality of a given correlation in an EPR scenario. Our approach is conceptually motivated by a particular perspective on EPR correlations arising from the language of causal modelling~\cite{pearl2000causality,Wood_2015}. In this approach, it is assumed that the relevant causal structure contains no directed causal influences between the parties that share the quantum system, which implies that `steering' is not a form of action-at-a-distance, but rather is a form of inference through a nonclassical common cause. 
This view leads to a `resource theory of steering' which is distinct from that proposed in previous works~\cite{gallego2015resource}. Our approach follows the recent developments of resource theories for studying nonclassicality of common-cause processes
wherein the free operations are local operations and shared randomness (LOSR)~\cite{de2014nonlocality, geller2014quantifying, gallego2017nonlocality, acin2012randomness, cowpie, schmid2020standard, schmid2020type, rosset2020type}.

We now motivate our approach and compare it to prior approaches. 
Readers already familiar with the conceptual framework of Refs.~\cite{cowpie,schmid2020standard} 
may wish to skip to Section~\ref{sec:intro-summary}, where we summarize our main technical results.

\subsection{Reimagining (and renaming) EPR `steering'}
Consider a scenario wherein two parties (Alice and Bob) share a bipartite quantum system prepared in an entangled state, and where Alice then chooses a measurement, performs it on her share of the system, and obtains an outcome. 
Conditioned on the outcome of the measurement, Alice updates her description of the quantum state that corresponds to Bob's system. Consequently, each of Alice's possible measurements  leads to an ensemble of potential updated states of Bob's system together with their associated probabilities of arising. The collection of these ensembles -- an `ensemble of ensembles' -- is termed an {\em assemblage}~\cite{pusey2013negativity}. Depending on {\em which} measurement she chooses to carry out, then, Alice chooses from which of these ensembles the quantum state of Bob's system will be drawn from. 

Schr{\"o}dinger considered this dependence of the wavefunction describing Bob's system on Alice's measurement 
choice to be `magic'~\cite{schrodinger1935discussion}, and termed this phenomena {\em quantum steering}. This terminology is motivated by the idea that Alice's choice of measurement exerts a {\em causal influence} on  Bob's system, even when the two are at an arbitrary distance. If one instead considers the quantum state to be merely a representation of one's {\em information} about a system, then Alice's ability to update this information conditioned on her measurement outcome is not in itself surprising. When Alice {\em learns} something about the true state of Bob's system, by virtue of her measurement on a system correlated with it, it is only natural that she would adjust her description of it (as a wavefunction), accordingly.  

 In this work we do not view `quantum steering' as evidence for nonlocal causal influences. Since the term `steering' has an intrinsic bias towards action-at-a-distance, in this manuscript we forgo its use in favor of more neutral terminology.\footnote{The importance of thinking clearly about the distinction between causation and inference has been argued (and demonstrated) in many previous works~\cite{schmid2020unscrambling,Jaynes,pearl2000causality,PhysRevA.86.012103,fuchs2002quantum,schmidinitial,schmid2021guiding}.} For example, we will refer to a `steering scenario' as an {\em EPR scenario}. 
 This attitude towards terminology follows the lines of Ref.~\cite{cowpie}, which avoided the use of the term `nonlocality' since it similarly suggests the existence of superluminal causal influences.  We will only use the term `steering' when referring to prior work, and even then we will write it between quotations marks.
 
With this in mind, without loss of generality, we recast the EPR scenario as follows. Firstly, the fundamental causal structure underpinning an EPR scenario is that depicted in Fig.~\ref{fig:bipartite-scenario-entangled}, wherein Alice and Bob are connected only by a {\em common cause}, with no causal influence between them. This is exactly the causal structure suggested by relativity theory. 
Second, we assume Alice has access to a classical input system and a classical output system, while Bob has access to a quantum output system \footnote{Considering Bob to have access to a quantum system can be thought of as Bob having a characterized (or `trusted') quantum device.}. 
A process with these two features defines an assemblage, and can be depicted as in Fig.~\ref{fig:bipartite-scenario-entangled}.

\begin{figure}[h!]
  \begin{center}
  \subcaptionbox{\label{fig:bipartite-scenario-entangled}}
{\put(-60,0){\includegraphics[width=0.18\textwidth]{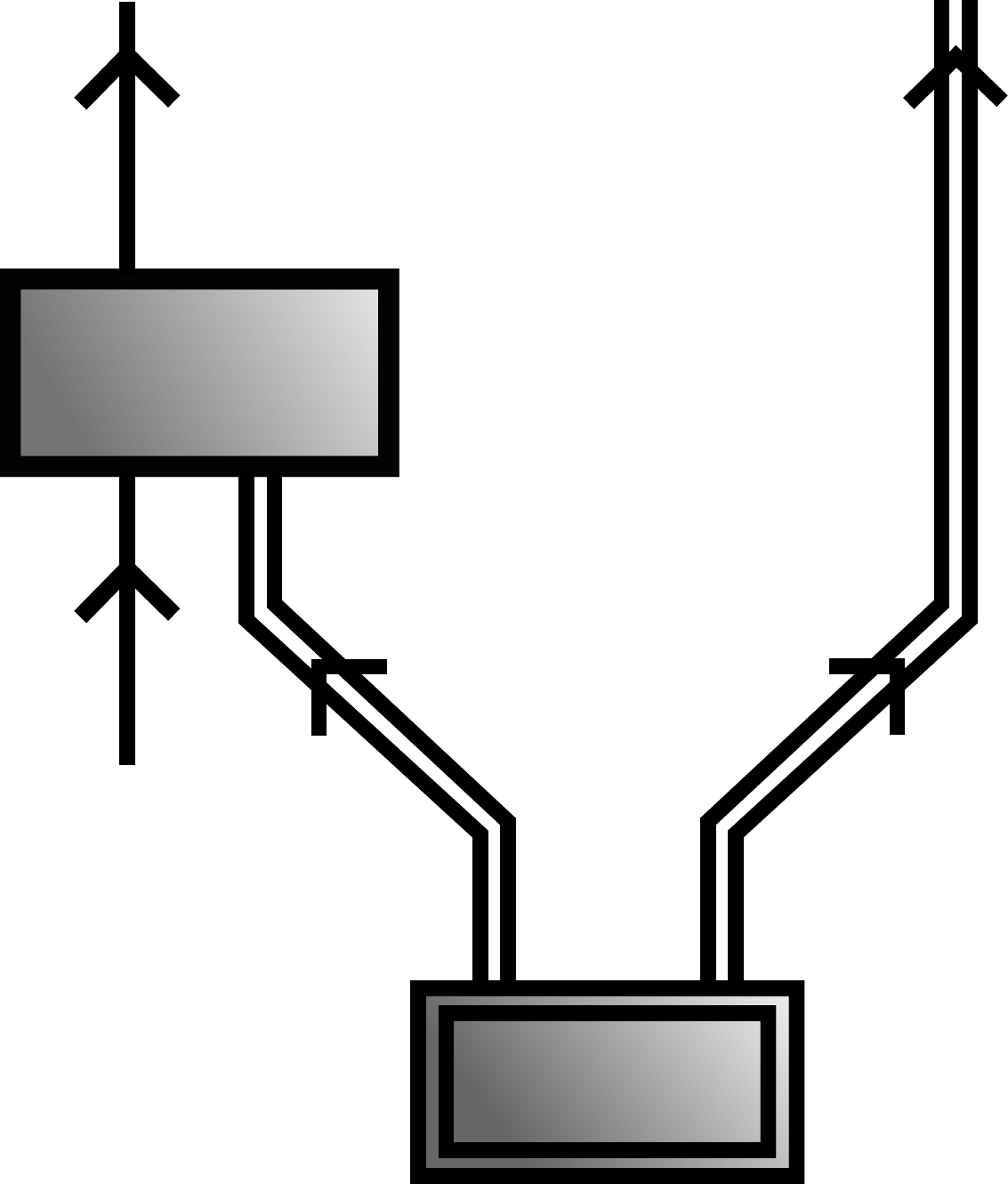}}
\put(-20,-20){$(a)$}
}
\hspace{80mm} 
  \subcaptionbox{\label{fig:bipartite-scenario-classical}}
{\put(-60,0){\includegraphics[width=0.19\textwidth]{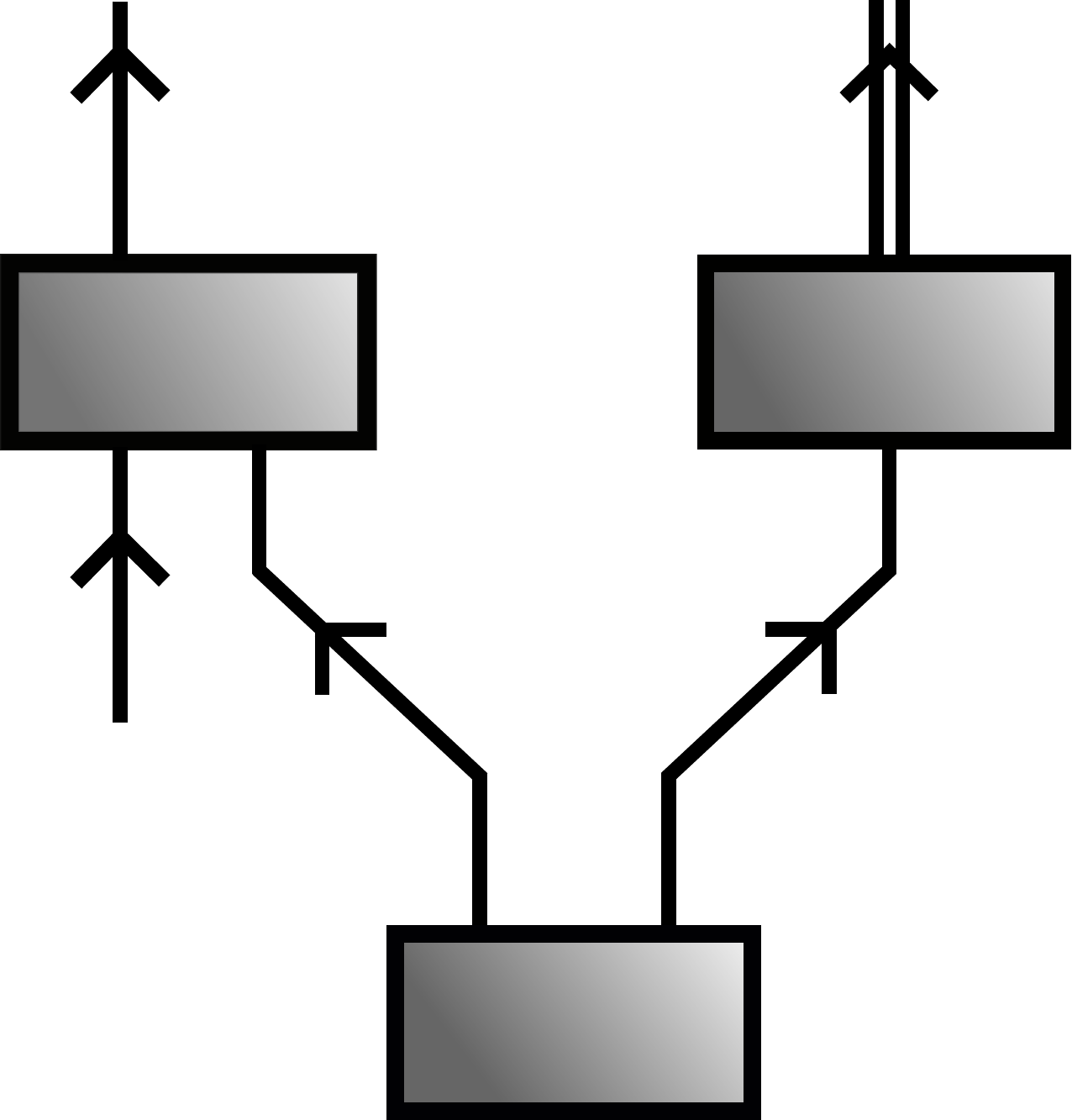}}
\put(-20,-20){$(b)$}
}

  \end{center}
  \caption{Depiction of a bipartite EPR scenario. Quantum systems are represented by double lines, while classical systems are depicted as single lines.  (a) A generic quantum assemblage requires Alice and Bob to share a quantum common cause.  (b) Free assemblages are those which can be generated when Alice and Bob share a classical common cause. }
\end{figure}

A fundamental question about assemblages in an EPR scenario is which assemblages admit a classical description, that is, one in terms of an underlying classical system that generates the observed correlations and that always takes a definite value. In the next section, we discuss that there exist assemblages which cannot be understood as being prepared with a shared classical random variable while setting up the relevant resource theory of nonclassicality in EPR scenarios.

\subsection{The resource theory of nonclassicality in EPR scenarios}

We follow the standard approach to quantifying the resourcefulness of given processes, namely, the framework of resource theories~\cite{coecke2016mathematical,chitambar2019quantum}.  In this approach, one determines two defining features of the resource theory: the \textit{enveloping theory}, which is the set of all possible resources one can produce in the setup; and the \textit{free subtheory} -- the set of operations on the resources that are considered to be freely available, as dictated by the physical constraints in the scenario under study. From this, one can directly determine the relative value of any pair of resources: if resource R can be converted to resource S, then R is at least as valuable as S, since (in the context of these free operations) R  can clearly be used for any purpose that S can be used for. 
 There are by now many useful resource theories, including those for entanglement (relative to local operations and classical communication~\cite{nielsen1999conditions,bennett1996concentrating,sanders2009necessary} or relative to local operations and shared randomness~\cite{buscemi2012all,schmid2020standard}), for nonclassicality in Bell scenarios~\cite{cowpie} and other common-cause processes~\cite{ schmid2020type,rosset2020type}, for post-quantumness~\cite{schmid2021postquantum,barrett2005nonlocal,brunner2009nonlocality}, for coherence~\cite{marvian2016quantify,marvian2016quantum,winter2016operational}, for athermality~\cite{brandao2013resource,skrzypczyk2014work,janzing2000thermodynamic,horodecki2013fundamental,gour2015resource,holmes2020quantifying}, and so on.

Historically, the EPR scenario has been viewed as a means of indirectly studying the properties of entangled states~\cite{wiseman2007steering}. More recently, researchers have begun to study assemblages as processes in and of themselves, since these contain all the relevant information for characterizing an EPR scenario. Here, we will follow this latter approach\footnote{To quantify resources in EPR scenarios, there are three particularly natural options. (i) One may quantify the value of bipartite quantum states, as these are necessary resources for achieving nonclassical assemblages; this has been studied in Ref.~\cite{schmid2020standard}. (ii) One may quantify the value of incompatible measurements, as these are also necessary resources for achieving nonclassical assemblages; this resource theory has been studied in Ref.~\cite{buscemi2020complete}. (iii) One may quantify the value of assemblages as resources in and of themselves. This is in some sense the most direct way of studying resources for nonclassicality, as motivated in the previous paragraph, and this is the sort of resource theory we develop in this paper.}, and so the resources in our resource theory are taken to be assemblages. 
We will only consider assemblages which can be realized within quantum theory (as opposed to those which can be realized with post-quantum resources~\cite{sainz2015postquantum,sainz2020bipartite,schmid2021postquantum}), and so we take our enveloping theory to be given by assemblages that can be generated by quantum common causes as in Fig.~\ref{fig:bipartite-scenario-entangled}.

The critical step in defining any resource theory is to determine the relevant set of free operations. 
It has previously been argued that the relevant free operations for studying nonclassicality in common-cause scenarios are given by the set of local operations and shared randomness~\cite{cowpie,schmid2020standard,schmid2020type,rosset2020type}. In this work, we also adopt this approach. In other words, free resources and free transformations on resources are those which can be generated by classical common causes. 
 In doing so, our approach unifies the study of `steering' resources with resources of `nonlocality' and (LOSR-)entanglement, showing that these are all simply different manifestations of nonclassicality of common-cause processes.

We now briefly reiterate the basic motivations for taking LOSR as free operations,  while specializing these arguments to the case of EPR scenarios. Our choice of free operations is guided by our assumptions about the causal structure, depicted in Fig.~\ref{fig:bipartite-scenario-entangled}. Firstly, we do not allow Alice or Bob to freely exert causal influences on one another, thus ruling out both classical and quantum communication. We place no limitations on the local quantum operations that each can carry out on the systems in their lab. Finally, as we wish to quantify the {\em non}classical properties of assemblages, we restrict the set of common causes which are taken to be free to be the {\em classical} common causes, so that a resource may be valuable only by virtue of having a nonclassical common cause. In other words, we take the set of free operations to include local operations and classical common causes. Since a more common term synonymous to `classical common cause' is `shared randomness', we refer to the set of free operations as {\em local operations and shared randomness}.  

From the perspective endorsed in this manuscript, then, the notion of resourcefulness embodied in a useful assemblage is {\em nonclassicality of its common-cause}. 
The free (classically explainable) assemblages are all and only those which can be generated by classical common causes, i.e., those of the form shown in Fig.~\ref{fig:bipartite-scenario-classical}. As we will see, this class of assemblages coincides exactly with the set of ``unsteerable'' assemblages---those that admit of so-called ``hidden-state models'', i.e., classical common-cause explanations. 
For such assemblages, Bob's system can always be viewed as locally prepared in a state that depends on a classical random variable which always takes a definite classical value, and Alice's choice of measurement can then be viewed as simply providing her information about this classical value, which in turn allows her to refine her description of the state of Bob's quantum system.

In contrast, the common cause generating any nonfree assemblages is necessarily nonclassical (i.e., cannot be viewed as a shared classical random variable), and so Alice's refinements of her knowledge about Bob's system are {\em not} compatible in this sense. However, there is currently no agreed-upon framework for formalizing Alice's refinements of her knowledge in this sense. Indeed, to fully understand how inferences are made through such a nonclassical common cause would presumably require a nonclassical generalization of the classical theory of inference. While some first attempts at such a program have been made \cite{fuchs2002quantum,leifer2013towards,schmid2020unscrambling}, arguably with some success when applied specifically to EPR scenarios, this ambitious problem remains very much unsolved.
A convenient feature of our resource theoretic approach is that it does not require us to resolve these difficult issues. This is because the structure of the resource theory is entirely determined by the free set of {\em classical} common-cause processes, which are well-understood even without such a novel theory of inference. 

\subsection{Comparison to prior work}

 A resource theory of `steering' in bipartite scenarios, distinct from ours, was defined a few years ago~\cite{gallego2015resource}. Therein, the free operations were taken to be (stochastic) local operations together with classical communication from Bob to Alice. This set of free operations was motivated by one particular application for which assemblages are known to serve as resources -- namely, this set is the most general one that does not compromise the security of one-sided device-independent quantum key distribution protocols.
In contrast, our resource theory is developed based on the natural physical limitations arising in any common-cause scenario. As such, our approach follows a recent body of work~\cite{cowpie,schmid2020standard,schmid2020type,rosset2020type} which demonstrates the importance of studying nonclassicality in common-cause scenarios (like the Bell scenario) within a resource-theoretic perspective underpinned by LOSR operations. 

We compare these two approaches further in Section~\ref{sec:relwork-comparison}. In particular, we note that our approach has the technical advantage that its set of free transformations is simpler to characterize and study.
More importantly, our resource theory  has the conceptual advantage that it allows for the unification of every type of nonclassical correlation in Bell-like scenarios, since all of these (including entanglement, `steering', and `nonlocality') can be viewed as resources of nonclassicality of common-cause processes. We expect this unification to be useful for better understanding the relationships between entanglement, EPR nonclassicality, and Bell nonclassicality, as well as for understanding the possibilities for interconversion between these (and many other~\cite{schmid2020type,rosset2020type}) forms of nonclassicality in common-cause scenarios. 
Furthermore, the principled approach we follow here to distill the essence of nonclassicality allows us to directly and uniquely generalize our framework to the case of multipartite EPR scenarios (which we provide a resource theory for in this work) and to Bob-with-input EPR scenarios~\cite{sainz2020bipartite} and channel EPR scenarios~\cite{piani2015channel}\footnote{Motivated by our common-cause description of an EPR scenario, we refer to so-called channel `steering' scenarios~\cite{piani2015channel} as channel EPR scenarios. Bob-with-input EPR scenarios~\cite{sainz2020bipartite} and channel EPR scenarios~\cite{piani2015channel} are generalizations of the traditional EPR scenario in which Bob is also allowed to have a (classical or quantum) input that further influences the state preparation of his quantum system.} (which we provide a resource theory for in a follow-up work).

\subsection{Summary of main results}\label{sec:intro-summary}

In this paper, we construct a resource theory of bipartite and multipartite `steering' under LOSR operations.  
This is the first time a resource-theoretic approach has been applied to the latter scenario.
Even within the traditional bipartite scenario, our approach differs from previous approaches to resource theories of `steering' \cite{gallego2015resource}, and clarifies some issues which arose  within them. 

This paper is divided into three main sections, where we discuss the definition and implications of the LOSR resource-theoretic approach to bipartite EPR scenarios (Section~\ref{sec:bipartite}), multipartite EPR scenarios (Section~\ref{sec:multi}), and the relation of this resource theory to previous work (Section~\ref{sec:relwork}).
In Sections~\ref{sec:bipartite} and~\ref{sec:multi}, we first recall the definition of the scenario of interest, and explicitly specify its most general LOSR processing of the corresponding EPR assemblage. Next, we show that resource conversion under free operations in the corresponding scenario can be evaluated with a single instance of a semidefinite program. We conclude each of these sections by highlighting properties of the pre-order of resources. In Section~\ref{sec:relwork}, we discuss the advantages of our LOSR approach in relation to the problem of choosing the set of free operations in a resource theory (Section~\ref{sec:relwork-comparison}) and consistently defining the set of free resources (Sections~\ref{sec:relwork-multi} and \ref{sec:relwork-exposure}).

\section{The bipartite EPR scenario}\label{sec:bipartite}

In this section, we define the resource theory of nonclassicality of bipartite assemblages under LOSR operations. We begin by formalizing the set of free (LOSR) operations. Then, we explicitly show that a resource conversion in this paradigm can be decided with a single instance of a semidefinite program. This is the first tool that allows one to systematically determine the relative nonclassicality of assemblages, that is, the convertibility relations that hold among them. We use this semidefinite program to study the possible conversions among a family of infinite assemblages. Moreover, we define new EPR measures and study the properties of the pre-order analytically. 

\subsection{Definition of the scenario and free assemblages}

The bipartite EPR scenario (see Fig.~\ref{fig:bipartite-scenario}) consists of two distant parties, Alice and Bob, that share a physical system and perform local actions on it. Alice performs (possibly incompatible) measurements on her subsystem: upon choosing a measurement setting denoted by $x$, she obtains a classical outcome $a$ with probability $p(a|x)$. We denote by $\X$ the set of classical labels that denote Alice's choice of measurement, and by $\A$ the set of labels for her measurement outcomes. Without loss of generality, we assume all the measurements to have the same outcome cardinality. By measuring her system, Alice refines her knowledge of Bob's system, which is now described by a conditional marginal state $\rho_{a|x}$. The relevant object of study is the ensemble of  ensembles of quantum states, dubbed an \emph{assemblage} \cite{pusey2013negativity}, that contains all the information characterizing an EPR scenario. It is defined as $\As_{\A|\X}=\{ \{\sigma_{a|x}\}_{a\in \A} \}_{ x \in \X}$, where each unnormalised state $\sigma_{a|x}$ is given by $\sigma_{a|x} := p(a|x) \rho_{a|x}$. That is, the probability that an updated state arises is given by the normalization factor for the corresponding state in the assemblage.

One can depict an assemblage as in Fig.~\ref{fig:bipartite-scenario}, where the classical and quantum systems are depicted with single and double lines, respectively. 
It is worth noticing that, even though we have chosen our enveloping theory to consist only of quantumly-realizable assemblages, the results derived in this section also apply to broader enveloping theories including those generated by arbitrary (e.g. post-quantum) common-causes. 
In fact, in the bipartite case, there is no difference between the assemblages\footnote{ 
It is worth noticing that in this definition of an EPR scenario -- which is just a generalisation of EPR's original thought experiment --  the system in Bob's laboratory remains quantum, or in other words, effectively admits a quantum description. Hence, even if post-quantum common causes are allowed, Bob's system does not require a post-quantum treatment. This is formally the same requirement as when Alice and Bob share a classical common cause: we still describe Bob's system with the language of quantum theory. Of course, one can define other types of experiments, where now Bob's system is allowed to be post-quantum.} that can be realized by quantum common causes and those that can be realized by arbitrary common causes. This follows from the GHJW theorem, proven by Gisin \cite{gisin1989stochastic} and Hughston, Jozsa, and Wootters \cite{hughston1993complete}, which (in our causal language) shows that {\em all} bipartite common-cause assemblages are quantumly-realizable.

\begin{figure}[h!]
  \begin{center}
  \subcaptionbox{\label{fig:bipartite-scenario}}
{\put(-60,40){\includegraphics[width=0.18\textwidth]{Figs/bipartite.pdf}}
\put(-45,75){$x$}
\put(-44,125){$a$}
\put(25,125){$\rho_{a|x}$}
\put(-20,0){$(a)$}
}
\hspace{80mm} 
\subcaptionbox{\label{fig:bipartite-scenario-LOSR}}
{\put(-85,0){\includegraphics[width=0.27\textwidth]{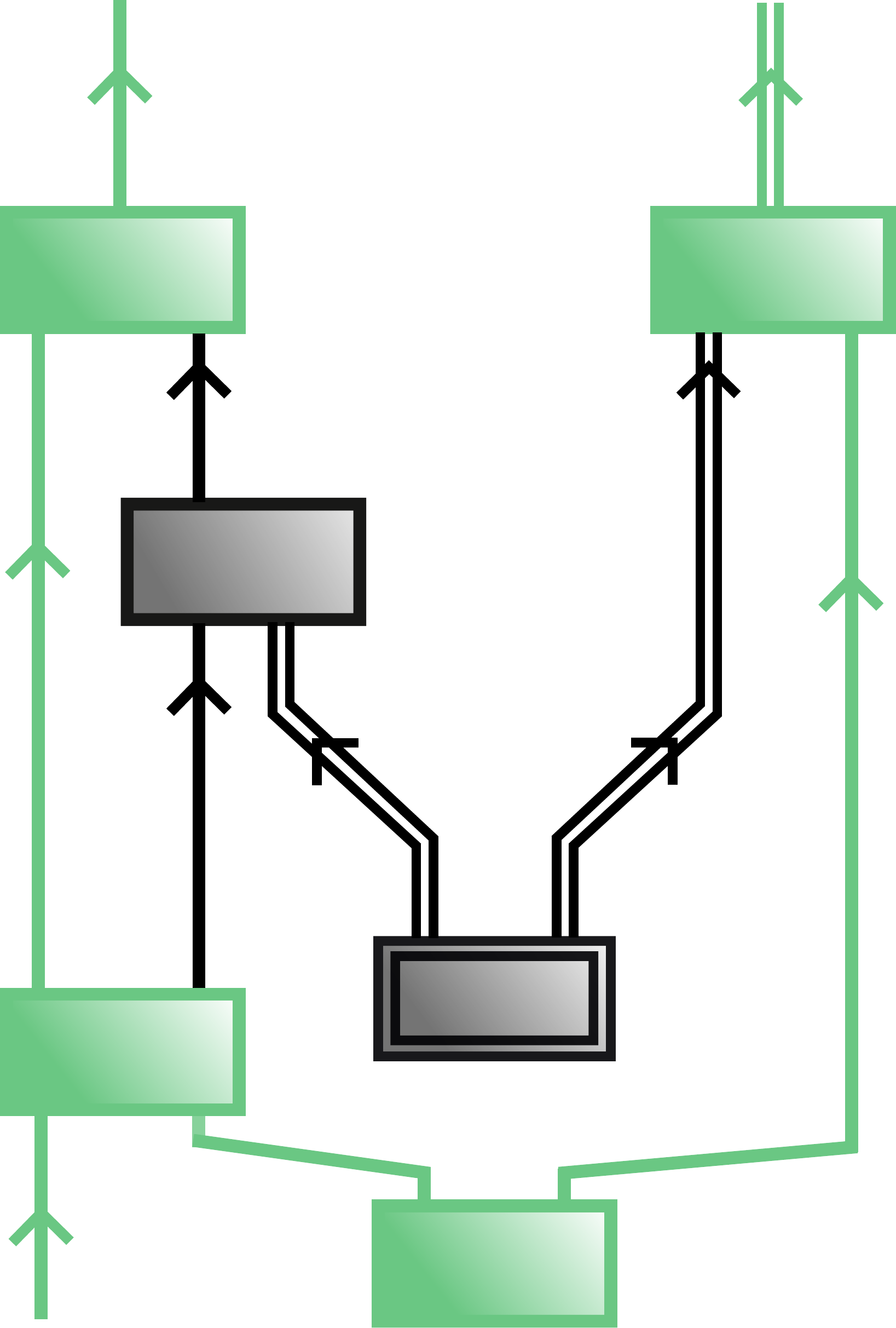}}
\put(-55,75){$x$}
\put(-55,120){$a$}
\put(-95,100){$c$}
\put(-21,5){$\lambda$}
\put(15,120){$\rho_{a|x}$}
\put(-72,10){$x^{\prime}$}
\put(-60,165){$a^{\prime}$}
\put(26,165){$\rho_{a^{\prime}|x^{\prime}}$}
\put(-25,-20){$(b)$}
}

  \end{center}
  \caption{Depiction of a bipartite EPR scenario. Quantum systems are represented by double lines, while classical systems are depicted as single lines.  (a) Quantum assemblage: Alice and Bob share a quantum common cause. (b) The most general LOSR operation on an assemblage in a bipartite EPR scenario.}
\end{figure}

In quantum theory, the state of the shared  system can  be  described by a density matrix $\rho$ and Alice implements generalised measurements, i.e., positive operator-valued measures (POVMs), which we denote by $\{\{M_{a|x}\}_{a\in \A}\}_{ x \in \X}$. In this case, the unnormalised states admit a quantum realisation of the form $\sigma_{a|x} = \text{tr}_A\{(M_{a|x} \otimes \id) \rho\}$. This is formalised as follows: 
\begin{defn}\textbf{ Quantumly-realizable assemblage.}\\
An assemblage $\As_{\A|\X}=\{ \{\sigma_{a|x}\}_{a\in \A} \}_{ x \in \X}$ has a quantum realisation iff there exists a Hilbert space $\cH_A$, a state $\rho$ in $\mathcal{H}_A \otimes \mathcal{H}_B$, and POVMs $\{\{M_{a|x}\}_{a\in \A}\}_{ x \in \X}$ on $\mathcal{H}_A$ such that 
\begin{align}\label{eq:quantum-bipartite}
\sigma_{a|x} = \mathrm{tr}_A\{(M_{a|x} \otimes \id) \rho\}
\end{align}
for all $x\in\X$ and $a\in\A$.
\label{def:Q-bipartite}
\end{defn}
\noindent

Finally, we consider the question of which assemblages admit a classical explanation: that is, which assemblages can be understood as a classical common-cause process. As argued in the introduction, these are all and only those that can be constructed freely from LOSR operations. 
Assemblages of this form are depicted in Fig.~\ref{fig:bipartite-scenario-classical}. 
Now, one can associate: 
\begin{compactitem}
\item a probability distribution $p(\lambda)$ with the box representing the state of the common cause (where we imagine that $\lambda$ is sampled according to this distribution, copied, and shared with both parties), 
\item a conditional probability distribution $p(a|x,\lambda)$ with the box representing the process by which Alice's outcome is generated (depending on her setting and the value of $\lambda$), 
\item a normalized quantum state $\rho_{\lambda}$ with the process by which Bob's quantum state is locally generated (depending on $\lambda$). 
\end{compactitem}
Hence, the elements of an LOSR-free assemblage can be written as $\sigma_{a|x}=\sum_{\lambda} p(\lambda) p(a|x,\lambda)  \rho_{\lambda}$, with $p(a|x,\lambda)$ being valid conditional probability distributions for all values of $\lambda$, which is the classical variable corresponding to the shared classical common cause. Traditionally, these LOSR-free assemblages have been referred to as `unsteerable’ assemblages and are mathematically expressed in terms of the so-called `local hidden states' $\sigma_\lambda := p(\lambda) \, \rho_\lambda$\footnote{ The `local hidden states’ terminology is motivated by the `local hidden variable’ terminology in Bell scenarios. In our view, this terminology is not ideal. Firstly, the states $\sigma_\lambda$ are local in the sense that Bob produces them locally, but the word local itself has a connotation in the context of `quantum nonlocality’ and so could be misleading. Secondly, it is irrelevant for the freeness of an assemblage whether or not the state describing Bob’s system is `hidden’. As such, it is better to refer to free resources as `classical common-cause resources', as this highlights precisely which feature of them is critical to their nonfreeness.}. 
 
\subsection{LOSR transformations between bipartite assemblages}

The most general LOSR transformation of an assemblage is illustrated in Fig.~\ref{fig:bipartite-scenario-LOSR}. It consists of a comb (which locally pre- and post-processes the classical variables in Alice's wing) \cite{chiribella2009theoretical} and a completely positive~\cite{NielsenChuang,schmidinitial} trace preserving (CPTP) map (which post-processes Bob's quantum system). Both operations are correlated by a classical variable $\lambda$.  Formally, a generic LOSR transformation illustrated in Fig.~\ref{fig:bipartite-scenario-LOSR} transforms one assemblage $\As_{\A|\X}$ into a new assemblage $\As'_{\A^{\prime}|\X^{\prime}}$ as follows:
\begin{align}\label{eq:LOSRtrans-bipartite} 
\sigma'_{a^\prime|x^\prime} = \sum_{\lambda,c}  \sum_{a,x} \, p(\lambda) \,  p(c,x|x^{\prime},\lambda) 
p(a|x) p(a^{\prime}|a,c)  \mathcal{E}_{\lambda}(\rho_{a|x})\,,
\end{align}
where 
\begin{compactitem}
\item $p(c,x|x^{\prime},\lambda) $ encodes the classical pre-processing of Alice's input $x$ as a function of $x^\prime$ and the shared classical randomness $\lambda$. Here $c$ denotes the variable to be transmitted through Alice's classical side channel toward the post-processing stage. 
\item $p(a^{\prime}|a,c) $ encodes the classical post-processing of Alice's output $a$, as a function of the classical information $c$ kept from the pre-processing stage. The output of the process is Alice's new outcome $a^\prime$. 
\item $ \mathcal{E}_{\lambda}[\cdot]$ is the CPTP map corresponding to Bob's local post-processing of his quantum system, as a function of the shared classical randomness $\lambda$.
\end{compactitem}

Notice that  if $\As_{\A|\X}$ is free, then $\As'_{\A^{\prime}|\X^{\prime}}$ is free as well -- that is, the set of free assemblages is closed under LOSR operations. This highlights the basic property of a resource theory: applying a free transformation on a free resource cannot create a resourceful object. 
Any assemblage that cannot be realised by local operations and shared randomness is nonfree and constitutes a resource of LOSR nonclassicality.

The local pre- and post-pocessings of Alice's classical variables are represented by the product $p(c,x|x^{\prime}, \lambda) \, p(a^{\prime}|a,c)$ in Eq.~\eqref{eq:LOSRtrans-bipartite}. Notice, however, that the classical variable $c$ is a function only of the classical values $x^{\prime}$ and $\lambda$, and that we have no constraints on the dimension of the classical variable $c$. Therefore, Eq.~\eqref{eq:LOSRtrans-bipartite} can be written as
\begin{align}\label{eq:LOSRtrans-bipartite2} 
\sigma'_{a^\prime|x^\prime} = \sum_{\lambda}  \sum_{a,x} \, p(\lambda) \,  p(a^{\prime}|a,x^{\prime}, \lambda)\, p(x|x^{\prime}, \lambda)\, p(a|x)\,  \mathcal{E}_{\lambda}(\rho_{a|x})\,.
\end{align} 
Notice that this expression satisfies the condition of no-retrocausation -- the variable $a$ cannot influence the value of the variable $x$.

A final remark pertains to a particular way to express a generic LOSR transformation, based on Fine's argument \cite{FinePRL} and discussed in Ref.~\cite{cowpie}. In the central point of this expression is the fact that the set of LOSR operations is convex and its extremal elements are deterministic \cite{cowpie}. This implies that Alice's pre- and post-processing can be decomposed as a convex combination of deterministic operations:
\begin{align}\label{eq:det-bipartite}
    p(a^{\prime}|a,x^{\prime}, \lambda)\, p(x|x^{\prime}, \lambda)= \sum_{\tilde{\lambda}} p(\tilde{\lambda}|\lambda) D(a^{\prime}|a,x^{\prime},\tilde{\lambda})\,D(x|x^{\prime},\tilde{\lambda}).
\end{align} 
Here, $D(a^{\prime}|a,x^{\prime},\tilde{\lambda})$ assigns a fixed outcome $a^{\prime}$ for each possible choice of $a$, $x^{\prime}$, and $\tilde{\lambda}$, i.e., $D(a^{\prime}|a,x^{\prime},\tilde{\lambda})=\delta_{a^{\prime},f_{\tilde{\lambda}}(a,x^{\prime})}$. Similarly, $D(x|x^{\prime},\tilde{\lambda})$ assigns a fixed outcome $x$ for each measurement $x^{\prime}$ and value of $\tilde{\lambda}$, i.e., $D(x|x^{\prime},\tilde{\lambda})=\delta_{x,g_{\tilde{\lambda}}(x^{\prime})}$. Let us define a new completely positive and trace non-increasing (CPTNI) map 
\begin{align}
\tilde{\mathcal{E}}_{\tilde{\lambda}}(\sigma_{a|x})= \sum_{\lambda} p(\lambda) p(\tilde{\lambda}|\lambda) \mathcal{E}_{\lambda}(p(a|x)\rho_{a|x}).
\end{align}
Notice that $\sum_{\tilde{\lambda}} \tilde{\mathcal{E}}_{\tilde{\lambda}}(\sigma_{a|x})$ forms a CPTP map. We are now in the position to rewrite Eq.~\eqref{eq:LOSRtrans-bipartite2} as follows
\begin{align}\label{eq:det-bipartite2}
     \sigma'_{a^{\prime}|x^{\prime}}= \sum_{\tilde{\lambda}}\sum_{a,x} D(a^{\prime}|a,x^{\prime},\tilde{\lambda})\,D(x|x^{\prime},\tilde{\lambda}) \tilde{\mathcal{E}}_{\tilde{\lambda}}(\sigma_{a|x}),
\end{align}
which yields the simplified characterisation of a generic LOSR transformation that we will use throughout.

\subsection{A semidefinite test for deciding resource conversions}\label{sec:sdp-bi}

An assemblage $\As_{\A|\X}$ can be converted into a different assemblage $\As'_{\A^{\prime}|\X^{\prime}}$ under LOSR operations if and only if there exist a collection of CPTNI maps $\tilde{\mathcal{E}}_{\lambda}$ such that $\As'_{\A^{\prime}|\X^{\prime}}$ can be decomposed as in Eq.~\eqref{eq:det-bipartite2}. Therefore, deciding whether $\As_{\A|\X}$ can be converted into $\As'_{\A^{\prime}|\X^{\prime}}$ under LOSR operations is equivalent to checking whether this decomposition is possible. We will now show that this can be decided with a single instance of a semidefinite program (SDP).

Recall that due to Choi-Jamiołkowski isomorphism~\cite{choi1975completely,jamiolkowski1972linear} every CPTP map $\mathcal{E}: \mathcal{H}_B \rightarrow \mathcal{H}_{B'}$ can be associated with an operator $W$ on $\mathcal{H}_B \otimes \mathcal{H}_{B'}$ such that $\mathcal{E}(\rho_B)=d_B\, \tr{W\,(\,\id_{B^{\prime}} \otimes \rho_{B}^{T})}{B}$, where $d_B$ is the dimension of the system $\rho_B$. Conversely, the operator $W$ can be written as $W=(\mathcal{E}\otimes\id_{B'})\ket{\Omega}\bra{\Omega}$, with $\ket{\Omega}=\frac{1}{\sqrt{d_B}}\sum_{i=1}^{d_B} \ket{ii}$. This isomorphism is crucial for reformulating our problem as an SDP. 
For Eq.~\eqref{eq:det-bipartite2} to hold, each $ \sigma'_{a^{\prime}|x^{\prime}}$ must admit the following decomposition
\begin{align}\label{eq:sigSDP-bipartite}
     \sigma'_{a^{\prime}|x^{\prime}}= \sum_{\lambda}\sum_{a,x}D(a^{\prime}|a,x^{\prime},\lambda)\,D(x|x^{\prime},\lambda)\,d_B\, \tr{W_{\lambda}\,(\,\id_{B^{\prime}} \otimes \sigma_{a|x}^{T})}{B}\,.
\end{align}
Here, the map $\tilde{\mathcal{E}}_{\lambda}(\sigma_{a|x})$ is written in the operator form with $W_{\lambda}$ being the Choi state, i.e., $\tilde{\mathcal{E}}_{\lambda}(\sigma_{a|x})=d_B\, \tr{W_{\lambda}\,(\,\id_{B^{\prime}} \otimes \sigma_{a|x}^{T})}{B}$, with $W_{\lambda}$ being a $(d_B \times d_{B’})$ by $(d_B \times d_{B’})$ matrix. Notice also that Eq.~\eqref{eq:sigSDP-bipartite} involves only finite sums, since the variable $\lambda$ enumerates the finitely-many deterministic distributions $D(a^{\prime}|a,x^{\prime},\lambda)$ and $D(x|x^{\prime},\lambda)$. The total number of the deterministic strategies encoded in $\lambda$ is equal to $|\A^{\prime}|^{|\A|\times|\X'|}\times|\X|^{|\X'|}$. 
This correspondence between elements $ \sigma'_{a^{\prime}|x^{\prime}}$ and $\sigma_{a|x}$ enables us to construct an SDP that checks whether $\As_{\A|\X}$ can be converted into $\As'_{\A^{\prime}|\X^{\prime}}$ under LOSR operations:

\begin{sdp}\label{SDP-bipartite}\textbf{$\As_{\A|\X} \, \overset{\text{LOSR}}{\longrightarrow} \, \As'_{\A^{\prime}|\X^{\prime}}$.}\\
The assemblage $\As_{\A|\X}$ can be converted into the assemblage $\As'_{\A^{\prime}|\X^{\prime}}$ under LOSR operations, denoted by $\As_{\A|\X} \, \overset{\text{LOSR}}{\longrightarrow} \, \As'_{\A^{\prime}|\X^{\prime}}$, if and only if the following SDP is feasible:

\begin{align}
\begin{split}
\textrm{given} \;\;\;& \{ \{\sigma_{a|x}\}_{a} \}_{ x}\,,\; \{ \{ \sigma'_{a^{\prime}|x^{\prime}}\}_{a^{\prime}}\}_{x^{\prime}}\,,\; \{D(a^{\prime}|a,x^{\prime},\lambda)\}_{\lambda,a^{\prime},a,x^{\prime}} \,, \; \{D(x|x^{\prime},\lambda)\}_{\lambda,x,x^{\prime}} \\
    \textrm{find} \;\;\;& \{(W_{\lambda})_{BB'} \}_{\lambda}  \\
    \textrm{s.t.} \;\;\;& \begin{cases} W_{\lambda} \geq 0\,,\\
      \tr{W_{\lambda}}{B^{\prime}} \propto \frac{1}{d}\,\id_B \;\;\;\; \forall \lambda\,, \\
      \sum_{\lambda} \tr{W_{\lambda}}{B^{\prime}} = \frac{1}{d}\, \id_B\,,\\
       \sigma'_{a^{\prime}|x^{\prime}}= \sum_{\lambda}\sum_{a,x} D(a^{\prime}|a,x^{\prime},\lambda)\, D(x|x^{\prime},\lambda)\,d_B\, \tr{W_{\lambda}\,(\,\id_{B^{\prime}} \otimes \sigma_{a|x}^{T})}{B}\,.
      \end{cases}
    \end{split}
\end{align}
When the conversion is not possible, we denote it by $\As_{\A|\X} \, \overset{\text{LOSR}}{\not\longrightarrow} \, \As'_{\A^{\prime}|\X^{\prime}}$.
\end{sdp}
SDP \ref{SDP-bipartite} is a feasibility problem, i.e., it checks if the feasible set is equal to the empty set. If this is the case, the primal optimal value is equal to $-\infty$, which means that there exists no set $\{W_{\lambda}\}_{\lambda}$ that satisfies the constraints specified by Eq.~\eqref{eq:det-bipartite2}. If the feasible set is not equal to the empty set, the optimal value is equal to zero, and the problem is feasible. Therefore, checking whether $\As_{\A|\X}$ can be converted into $\As'_{\A^{\prime}|\X^{\prime}}$ under LOSR operations requires a single instance of an SDP.

\subsection{Properties of the pre-order}
When one assemblage can be freely converted into another, then the former is said to be \emph{at least as nonclassical} as the latter.
Therefore, the pre-order of assemblages gives information about their relative nonclassicality. In this section, we study the nonclassicality among members of an infinite family of assemblages defined below.

Consider an EPR scenario where $\A=\X=\{0,1\}$, and Bob's dimension is $2$. Imagine Alice and Bob share an entangled state of the form $\ket{\theta}=\cos{\theta} \ket{00} + \sin{\theta} \ket{11}$, where $\theta \in \left(0,\sfrac{\pi}{4}\right]$. The measurements that Alice performs on her share of the entangled state are given by $\widetilde{M}_{a|0} = \frac{1}{2}\{\id + (-1)^a \sigma_z\}\,$ and $\widetilde{M}_{a|1} = \frac{1}{2}\{\id + (-1)^a \sigma_x\}$, with $\sigma_z$ and $\sigma_x$ being Pauli matrices. Then, the assemblage elements can be written as
\begin{align}\label{eq:ass_elem}
\sigma^\theta_{a|x} = \tr{\widetilde{M}_{a|x} \otimes \id \ket{\theta}\bra{\theta}}{\mathrm{A}}.
\end{align}
This infinite family of assemblages is indexed by one parameter -- the angle $\{\theta\}$. We will now introduce one more parameter to this family. The new parameter $\{p\}$ is responsible for mixing the elements $\sigma^\theta_{a|x}$ with noise. Let us define a family of assemblages $\FA$ as: 
\begin{align}\label{eq:thefam}
\FA &= \left\{ \As^{\theta,p}_{\A|\X} \, \Big\vert \, \theta \in \left(0,\sfrac{\pi}{4}\right], p \in [0,1] \right\}\,,\\ \nonumber
\text{where} \quad \As^{\theta,p}_{\A|\X} &= \left\{ \left\{p \, \sigma^\theta_{a|x} + (1-p) \, \frac{\id}{4} \right\}_{a\in \A}\right\}_{ x \in \X} \,.
\end{align} 

This family of assemblages has an infinite number of elements. Each element is indexed by two parameters -- the angle and the probability, $\{\theta,p\}$ respectively. We will sometimes focus on a subset of $\FA$ with a fixed value of $p$. In such cases, for $p=\varrho$, we define $\FA_\varrho=\left\{ \As^{\theta,\varrho}_{\A|\X} \, \Big\vert \, \theta \in \left(0,\sfrac{\pi}{4}\right], p=\varrho \right\}$. For example, for $p=1$, we have the following family of assemblages:
\begin{align}\label{eq:thefam1}
\FA_1 &= \left\{ \As^{\theta,1}_{\A|\X} \, \Big\vert \, \theta \in \left(0,\sfrac{\pi}{4}\right], p =1 \right\}\,,\\ \nonumber
\text{where} \quad \As^{\theta,1}_{\A|\X} &= \left\{ \left\{ \sigma^\theta_{a|x}  \right\}_{a\in \A} \right\} _{ x \in \X} \,.
\end{align}
For simplicity, we denote $\As^{\theta,p}_{\A|\X}=\As^{\theta,p}$ in this section. 

We now study the properties of the pre-order among the resources in $\FA$. First, we run the SDP \ref{SDP-bipartite} and test which conversions between resources parametrized by different values of $\{\theta,p\}$ are possible. Second, in Section~\ref{sec:new-monotones}, we focus on $\FA_1$ and confirm the results obtained from the SDP analytically.

To test whether two assemblages in $\FA$ can be converted into each other, we check if they satisfy the constraints specified by Eq.~\eqref{SDP-bipartite}. The solutions to the SDP (computed in Matlab \cite{MATLAB:2010}, using the software CVX \cite{grant2013cvx,blondel2008recent}, the solver SDPT3 \cite{tutuncu1999sdpt3} and the toolbox QETLAB \cite{qetlab}; see the code at \cite{github}) are illustrated in Fig.~\ref{fig:SDP-bipartite}. In this figure, each dot represents one assemblage $\As^{\theta,p}$. For example, the point indexed by $\{ \theta=\pi/6, p=0.9 \}$ corresponds to $\As^{\pi/6,0.9} = \left\{\left\{0.9 \, \sigma^{\pi/6}_{a|x} + 0.1 \, \frac{\id}{4} \right\}_{a\in \A}\right\}_{x \in \X}$. 
The black dots correspond to nonfree assemblages, and the grey dot represents a free assemblage. 
Note that the assemblage $\As^{\pi/4,1}$ is one obtained when Alice and Bob share a quantum system prepared in a `maximally entangled'\footnote{Here, by a `maximally entangled' quantum state we mean one where entanglement is quantified as per a resource theory based on local operations and classical communication as the free operations \cite{LOCCentang,schmid2020standard}, which is the standard approach. In the case of two qubits this may be the singlet state.} state, and Alice performs two suitable Pauli measurements in her share of the system. The arrows represent possible conversions, e.g., the arrow pointing from $\As^{\pi/4,1}$ to $\As^{\pi/6,0.8}$ means that $\As^{\pi/4,1}$ can be converted into $\As^{\pi/6,0.8}$ under LOSR operations. The grey, dashed lines represent trivial conversions. For the sake of simplicity, Fig.~\ref{fig:SDP-bipartite} represents only nine assemblages, which is already sufficient to illustrate some interesting features of the pre-order among assemblages in $\FA$. 

\begin{figure}
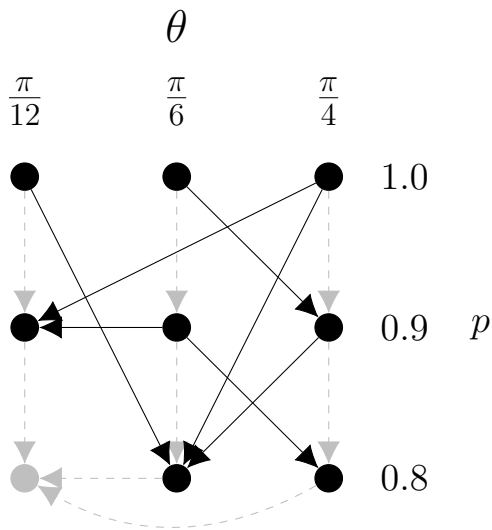

 \centering
  \include{Figs/SDP-results}
  \caption{Possible conversions between elements of $\FA$. The black dots represent the assemblages $\As^{\theta,p}$. The arrows represent possible conversions.}
  \label{fig:SDP-bipartite}
\end{figure}

Fig.~\ref{fig:SDP-bipartite} displays some elements of $\FA_1$ that are unordered in our LOSR resource theory of assemblages; e.g., notice that no conversions are possible among assemblages where $p=1$. In contrast, some conversions {\em are} possible among assemblages for which $p=0.9$, and among assemblages for which $p=0.8$.
Moreover, this figure is not symmetric: conversions of assemblages with higher value of $\theta$ to those with lower value of $\theta$ are more common. For example, $\As^{\pi/4,1}$ can be converted into $\As^{\pi/12,0.9}$, but $\As^{\pi/12,1}$ cannot be converted into $\As^{\pi/4,0.9}$. Finally, all arrows point in one direction only, suggesting that there are no equivalent assemblages in this family. 

Notice that we used the SDP \ref{SDP-bipartite} to study the pre-order of assemblages in $\FA$, where $\A=\X=\{0,1\}$, and Bob's dimension is $2$. However, the SDP is not constrained to such simple examples. In principle, resource conversion can be evaluated via the SDP \ref{SDP-bipartite} for families of assemblages with an arbitrary cardinality of $\A$ and $\X$, and arbitrary Bob's dimension.

\subsubsection{EPR monotones}\label{sec:oldmonotones}

A common approach to unraveling the pre-order of resources in a resource theory is by employing so-called \textit{resource monotones}. Formally, a resource monotone is a function which is monotonic under free
operations. Resource monotones enable comparison of resources and identification of conversion relations and equivalence classes. 
That is, if the real numbers that resource monotone $M$ assigns to the pair of resources $R_1$ and $R_2$ satisfy $M(R_1) < M(R_2)$, then one can conclude that the monotone $M$ witnesses that $R_1 \not\longrightarrow R_2$. Strictly speaking, hence, a resource monotone usually gives only partial information about the pre-order of resources, and one generically may need a collection of resource monotones (called a \textit{complete set} of them) to have the complete information to fully specify the pre-order of resources based on them. Nevertheless, resource monotones are friendly ways to explore the pre-order of resources, and the real numbers in their co-domain are interpreted as a (possibly incomplete) \textit{quantification} of the nonfreeness of the resources in its domain.  

In the particular case of our resource theory, an EPR monotone is a function from the space of assemblages into real numbers, whose value does not increase under LOSR operations. 
Among the existing measures of bipartite assemblages, the \textit{`steerable' weight} \cite{weight}, \textit{robustness of 'steering'} \cite{robustness}, and \textit{relative entropy} \cite{gallego2015resource} were shown to be EPR monotones in a resource theory of `steering' with stochastic local operations assisted by one-way classical communication (S-1W-LOCC) being the set of free operations \cite{gallego2015resource}. The set of S-1W-LOCC operations allows, in particular, for one party to generate randomness and share it with the second party by classical communication; hence, the set of S-1W-LOCC operations strictly includes the set of LOSR operations. It follows that all EPR monotones defined for S-1W-LOCC are therefore also monotones for LOSR. It has moreover been shown that the first two of these monotones -- `steerable' weight and robustness of `steering' -- can be reformulated in a type-independent form, i.e., they can be used to compare the LOSR-nonclassicality of resources
of arbitrary types (where an EPR assemblage is one specific type of a resource) \cite{rosset2020type}. Let us now recall their formal definitions for the case of study in this section.

\begin{defn}[`Steerable' weight \cite{weight}]
The `steerable' weight of an assemblage $\As_{\A|\X}$ is given by the minimum $\mu$ such that $\As_{\A|\X}$ can be decomposed as
\begin{align}
    \As_{\A|\X} = \mu \, \As_{\A|\X}^{S} + (1-\mu) \, \As_{\A|\X}^{\mathrm{free}},
\end{align}
where $\As_{\A|\X}^{S}$ is an arbitrary nonfree assemblage and $\As_{\A|\X}^{\mathrm{free}}$ is an arbitrary free assemblage. 
\label{def:monotone_weight}
\end{defn}

We provide the following intuition for this definition. Imagine Alice wants to prepare a given assemblage $\As_{\A|\X}$ using the minimal amount of a nonfree resource. To achieve her goal, she prepares a free assemblage $\As_{\A|\X}^{\mathrm{free}}$ most of the time (that is, a $1-\mu$ fraction of the rounds), and sometimes (a $\mu$ fraction of the rounds), she prepares a nonfree assemblage $\As_{\A|\X}^{S}$. On average, hence she prepares $\As_{\A|\X}$. The `steerable' weight quantifies the minimal amount of $\As_{\A|\X}^{S}$ needed to generate the desired assemblage in this way. 

\begin{defn}[Robustness of `steering' \cite{robustness}]
The `steering' robustness of an assemblage $\As_{\A|\X}$ is given by the minimum $\nu$ such that the assemblage
\begin{align}
    \As_{\A|\X}^{\mathrm{free}} = \frac{1}{1+\nu} \, \As_{\A|\X} + \frac{\nu}{1+\nu} \, \As_{\A|\X}^{\prime}
\end{align}
is free, with $\As_{\A|\X}^{\prime}$ being an arbitrary assemblage.
\label{def:monotone_robustness}
\end{defn}

Robustness of `steering' can be understood as follows. Imagine Alice holds an assemblage $\As_{\A|\X}$ and she wants to make it free by mixing it with some other assemblage $\As_{\A|\X}^{\prime}$. Robustness of `steering' quantifies the minimal amount of mixing needed to make $\As_{\A|\X}$ a free assemblage. If one optimises $\As_{\A|\X}^{\prime}$ over all assemblages, this leads to mixing $\As_{\A|\X}$ with worst-case noise. If one restricts the type of noise that can be added to the original assemblage, this measure is referred to as generalized robustness or random robustness, depending on the type of noise added. 

We conjecture that relative entropy could be defined in a type-independent way as well, as similar measures exist for states and Bell scenarios \cite{gallego2017nonlocality}.

\subsubsection{New EPR monotones}\label{sec:new-monotones}
We now develop a method for obtaining EPR monotones from Bell inequalities and use it to construct a family of monotones that certify the incomparability of elements of $\FA_1$. We only focus on $\FA_1$ in this section, hence we drop the index $p=1$ and denote $\As^{\theta,1}=\As^{\theta}$. 

To prove that the elements of $\FA_1$ are unordered as per the LOSR resource theory of assemblages, it suffices to find a set of EPR LOSR monotones $\FM = \{\M_j\}$ such that, for every pair $(\theta_1,\theta_2)$ there exists a pair $(\M_{\theta_1},\M_{\theta_2})$ with the following properties: 
\begin{align}\label{eq:monotones_condition}
\begin{cases}
\M_{\theta_1}(\As^{\theta_1}) > \M_{\theta_1}(\As^{\theta_2}) \quad \text{-- which implies } \As^{\theta_2} \not\xrightarrow{\mathrm{LOSR}} \As^{\theta_1} \,,\\
\M_{\theta_2}(\As^{\theta_1}) < \M_{\theta_2}(\As^{\theta_2}) \quad \text{-- which implies } \As^{\theta_1} \not\xrightarrow{\mathrm{LOSR}} \As^{\theta_2} \,.
\end{cases}
\end{align} 
To construct such a family of monotones, we will make use of the EPR inequalities~\cite{cavalcanti2009experimental} constructed from the Bell inequalities presented in Ref.~\cite[Eq.~(1)]{bamps2015sum}. These Bell inequalities, which we recall in the Appendix \ref{app:bipartite-Bell}, are uniquely maximized by the states $\ket{\theta}$ (via a self-testing result~\cite{bamps2015sum}). In the Appendix \ref{app:steeringin-bipartite}, we use these inequalities to construct an EPR functional, which we denote $S_\eta[\As]$. This functional defines a family of EPR inequalities that are uniquely maximized by assemblages living in $\FA_1$. For a formalization of this statement and definition of $S_\eta[\As]$, see Appendix \ref{app:steeringin-bipartite} and Definition \ref{def:SFT} therein.

We construct the monotone $\M_\eta$ from the EPR functional $S_\eta[\As]$ following a yield-based construction:

\begin{defn}\label{def:Meta} The resource monotone $\M_\eta$, for $\eta \in (0,\sfrac{\pi}{4}]$, is defined as
\begin{align}
\M_\eta[\As] := \max_{\widetilde{\As}} \{ S_\eta[\widetilde{\As}]: \As \xrightarrow{\mathrm{LOSR}} \widetilde{\As} \}.
\end{align}
\end{defn}
We will now show that each of the monotones $\M_\eta[\As]$, when evaluated on assemblages in  $\FA_1$, is uniquely maximized by $\As^{\eta}$. First, note that
\begin{align}\label{eq:property1}
\M_\eta[\As] \leq S_\eta[\As^{\eta}] \quad \forall \, \eta \in \left(0,\sfrac{\pi}{4}\right]\,,
\end{align}
since all assemblages in this traditional bipartite scenario admit a quantum realisation. Hence, $S_\eta[\widetilde{\As}]$ is upperbounded by its maximum quantum violation, which is given by $S_\eta[\As^{\eta}]$. Moreover,
\begin{align} \label{eq:property2}
\M_\eta[\As^{\eta}] = S_\eta[\As^{\eta}] \,.
\end{align}
We now show that that equality in Eq.~\eqref{eq:property2} only holds when $\M_\eta$ and $S_\eta$ are evaluated on the same assemblage. 

\begin{thm}\label{thm:4}
Let $\M_\eta$ be an EPR monotone from Definition~\ref{def:Meta} and $S_\eta$ be an EPR functional given in Definition~\ref{def:SFT}. Then, if $\theta \neq \eta$, $\M_\eta[\As^{\theta}] < S_\eta[\As^{\eta}]$.
\end{thm}
\begin{proof}
Let us prove this by contradiction. Our starting assumption is that there exists a pair $(\theta\,,\, \eta)$ with $\theta \neq \eta$, such that $\M_\eta[\As^{\theta}] = S_\eta[\As^{\eta}]$. Then, one of the two should happen: 

\bigskip

\noindent \underline{First case:}  $\M_\eta[\As^{\theta}] = S_\eta[\As^{\theta}]$. \\
In this case, the solution to the maximisation problem in the computation of $\M_\eta$ is achieved  by $\As^{\theta}$ itself. \\
Our starting assumption then tells us that $S_\eta[\As^{\theta}] = S_\eta[\As^{\eta}]=S_\eta^{\max}$.\\
From Remark \ref{rem:iff} (see Appendix~\ref{app:steeringin-bipartite}) it follows that necessarily $\theta=\eta$, which contradicts our initial condition. 

\bigskip

\noindent \underline{Second case:}  $\M_\eta[\As^{\theta}] = S_\eta[\widetilde{\As}]$, with $\As^{\theta} \xrightarrow{\mathrm{LOSR}} \widetilde{\As}$ . \\
In this case, the solution to the maximisation problem in the computation of $\M_\eta$ is achieved by an LOSR processing of $\As^{\theta}$. \\
Our starting assumption then tells us that $S_\eta[\widetilde{\As}] = S_\eta[\As^{\eta}]=S_\eta^{\max}$.\\
Let $\rho$ and $\{\{M_{a|x}\}_{a\in A}\}_{x\in X}$ be any quantum state and measurements that realise the quantum assemblage $\widetilde{\As}$. From Remark \ref{rem:thelast} (see Appendix~\ref{app:steeringin-bipartite}), we know that $\rho$ is equivalent to $\ket{\eta}$ up to local isometries. But since local isometries are free LOSR operations, this means that $\ket{\theta}$ is more LOSR-entangled\footnote{We say that a state $\ket{\theta}$ is more LOSR-entangled than a different state $\ket{\eta}$ if $\ket{\theta}$ can be converted freely to $\ket{\eta}$ with LOSR operations.} than $\ket{\eta}$, which contradicts the result of Ref.~\cite{schmid2020standard} that all two-qubit pure entangled states are LOSR-inequivalent. 

\end{proof}

We can now prove that $\FA_1$ is composed of pairwise unordered resources by identifying a pair of monotones that satisfies Eqs.~\eqref{eq:monotones_condition}. As we see next, this is achieved by choosing $\M_j = \M_{\theta_j}$.

\begin{thm}\label{thm:M-bipartite}
For every pair $\theta_1 \neq \theta_2$, the monotones $M_{\theta_1}$ and $M_{\theta_2}$ given by Definition \ref{def:Meta} satisfy 
\begin{align} \label{eq:6}
\M_{\theta_1}(\As^{\theta_1}) &> \M_{\theta_1}(\As^{\theta_2}) \,, \\
\M_{\theta_2}(\As^{\theta_1}) &< \M_{\theta_2}(\As^{\theta_2}) \,. \label{eq:7}
\end{align}
\end{thm}
\begin{proof}
Let us first prove Eq.~\eqref{eq:6}. On the one hand, $\M_{\theta_1}(\As^{\theta_1}) = S_{\theta_1}[\As^{\theta_1}]$. On the other hand, since $\theta_1 \neq \theta_2$, Theorem \ref{thm:4} implies that $\M_{\theta_1}(\As^{\theta_2}) < S_{\theta_1}[\As^{\theta_1}]$. Therefore,  Eq.~\eqref{eq:6} follows. \\
The proof of Eq.~\eqref{eq:7} follows similarly. 
\end{proof}

\begin{cor}\label{cor:bip}
The infinite family of EPR monotones $\FM = \{\M_\eta \, \vert \, \eta \in (0,\sfrac{\pi}{4}] \}$ certifies that the infinite family of assemblages $\FA_1$ is composed of pairwise unordered resources. 
\end{cor}

We defined new measures of EPR assemblages and showed that they certify that the elements of $\FA_1$ are unordered in the resource theory of assemblages under LOSR operations. As we showed using the SDP \ref{SDP-bipartite}, this property of the pre-order appears to be unique to assemblages $\As^{\theta,p}$ with $p=1$.

\section{The multipartite EPR scenario}\label{sec:multi}

We now develop a resource theory for multipartite EPR scenarios under LOSR operations. This is the first time multipartite EPR scenarios have been studied in a resource-theoretic framework. 
Similarly to the bipartite case, we show that resource conversion can be decided with a single instance of an SDP in this paradigm. Moreover, we define new measures of multipartite nonclassicality and use them to analytically study the properties of the pre-order. Our results show that multipartite scenarios are easier and more natural to understand from an LOSR perspective rather than an LOCC one; 
we elaborate on this in Section~\ref{sec:relwork-comparison}.

In this paper, we focus on one particular type of multipartite EPR scenario with multiple Alices and one Bob. However, the number of quantum parties (Bobs) can also be increased. Indeed, one possible multipartite setup consists of a single Alice and multiple Bobs. In such a scenario, results from bipartite scenarios may sometimes directly generalise. However, this generalisation might not be straight-forward and a detailed analysis of this multipartite scenario is beyond the scope of this work. For a description of a scenario with multiple Bobs, we refer the reader to Refs.~\cite{cavalcanti2011unified,cavalcanti2015detection}.
Focusing on scenarios with more-than-one Alice is however necessary when one wants to allow for post-quantum assemblages, which is crucial when exploring ways to single out quantum phenomena from basic principles -- a.k.a.~studying quantum `from the outside'.

\subsection{Definition of the scenario  and free assemblages}

In the multipartite EPR scenario of interest to us, $k+1$ separated parties share a physical system. In analogy to the bipartite scenario, we consider $k$ parties called Alices that hold measurement devices. Each Alice, labeled by $A_{i \in \{1...k\}}$, decides on a classical input $x_i$ from the set $x_i\in \{1,...,m_A\} =: \X$ and generates a classical outcome $a_i$ from the set $a_i \in \{0,...,o_A-1\} =: \A$ with probability $p^{i}(a_i|x_i)$. Without loss of generality,  we will take the sets $\X$ and $\A$ to be the same for all the Alices. When Alices perform the measurements on their subsystems, they update their knowledge about Bob's subsystem which is now described by a conditional marginal state $\rho_{a_1...a_k|x_1...x_k}$, which depends on all Alices inputs and outputs. In this scenario, the elements of the assemblage $\As_{\A_1\ldots\A_k|\X_1\ldots\X_k} = \{\{\sigma_{a_1...a_k|x_1...x_k}\}_{a_1...a_k}\}_{ x_1...x_k}$ are given by $\sigma_{a_1...a_k|x_1...x_k} := p(a_1...a_k|x_1...x_k) \rho_{a_1...a_k|x_1...x_k}$. The multipartite EPR scenario is illustrated in Fig.~\ref{fig:multipartite-scenario} for $k=2$ (two Alices, one Bob). 

In this paper, we are interested in studying the common-cause processes between multiple parties, hence we focus on the global object given by the assemblage, and not on the individual agents. Nevertheless, we would like to note that it is always possible to recover the individual parties by tracing out the irrelevant subsystems. It would be interesting to study what inferences each Alice can make about the other parties. However, this problem is related to the quantum marginal problem~\cite{klyachko2006quantum} and is beyond the scope of our framework. 

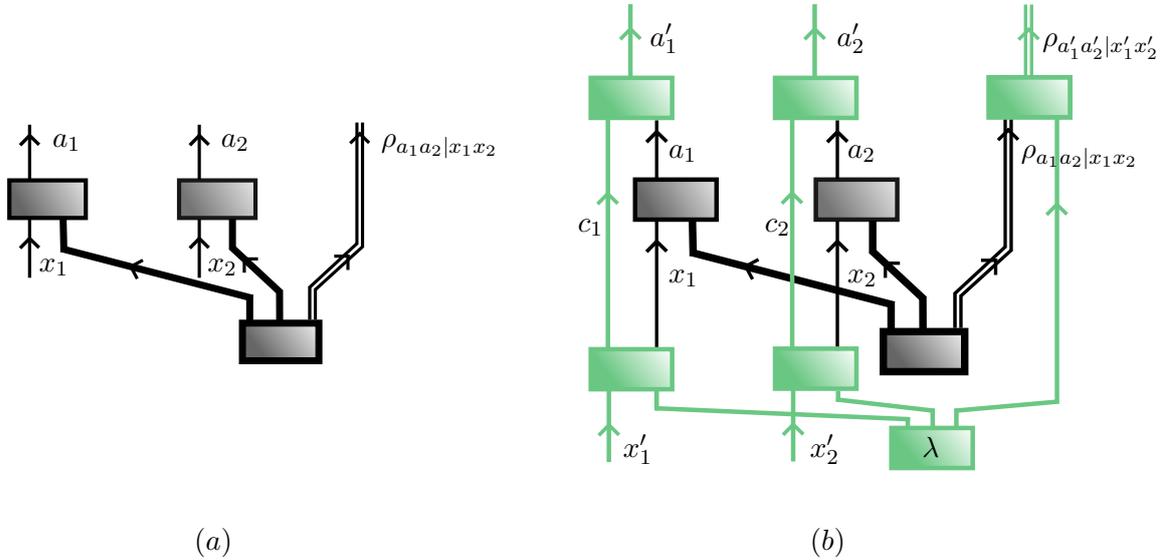
\begin{figure}[h!]
  \begin{center}
  \subcaptionbox{\label{fig:multipartite-scenario}}
{\put(-70,70){\includegraphics[width=0.3\textwidth]{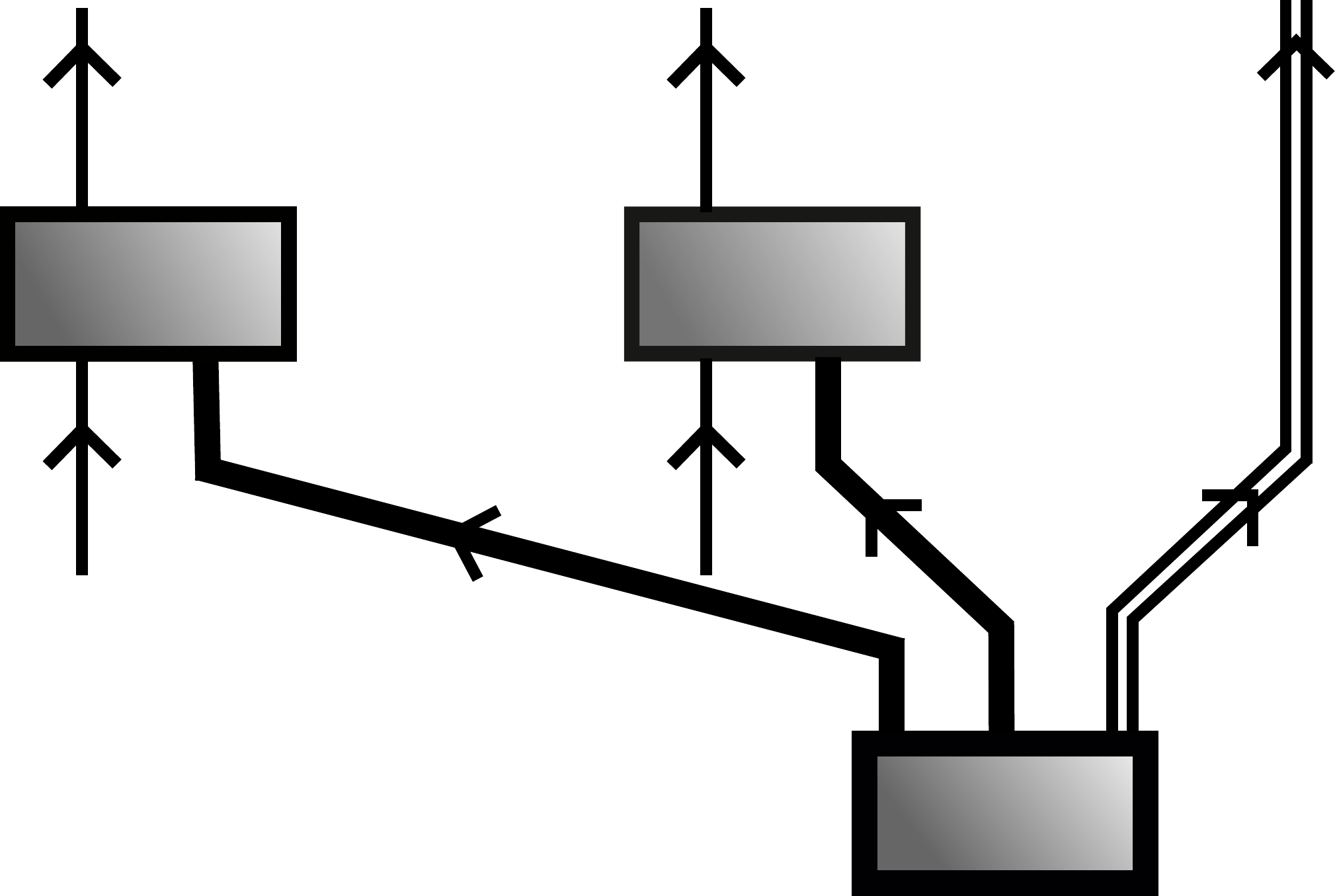}}
\put(-58,105){$x_1$}
\put(-53,152){$a_1$}
\put(5,105){$x_2$}
\put(10,152){$a_2$}
\put(70,152){$\rho_{a_1a_2|x_1x_2}$}
\put(0,0){$(a)$}
}
\hspace{95mm} 
\subcaptionbox{\label{fig:multipartite-scenario-LOSR}}
{\put(-130,0){\includegraphics[width=0.4\textwidth]{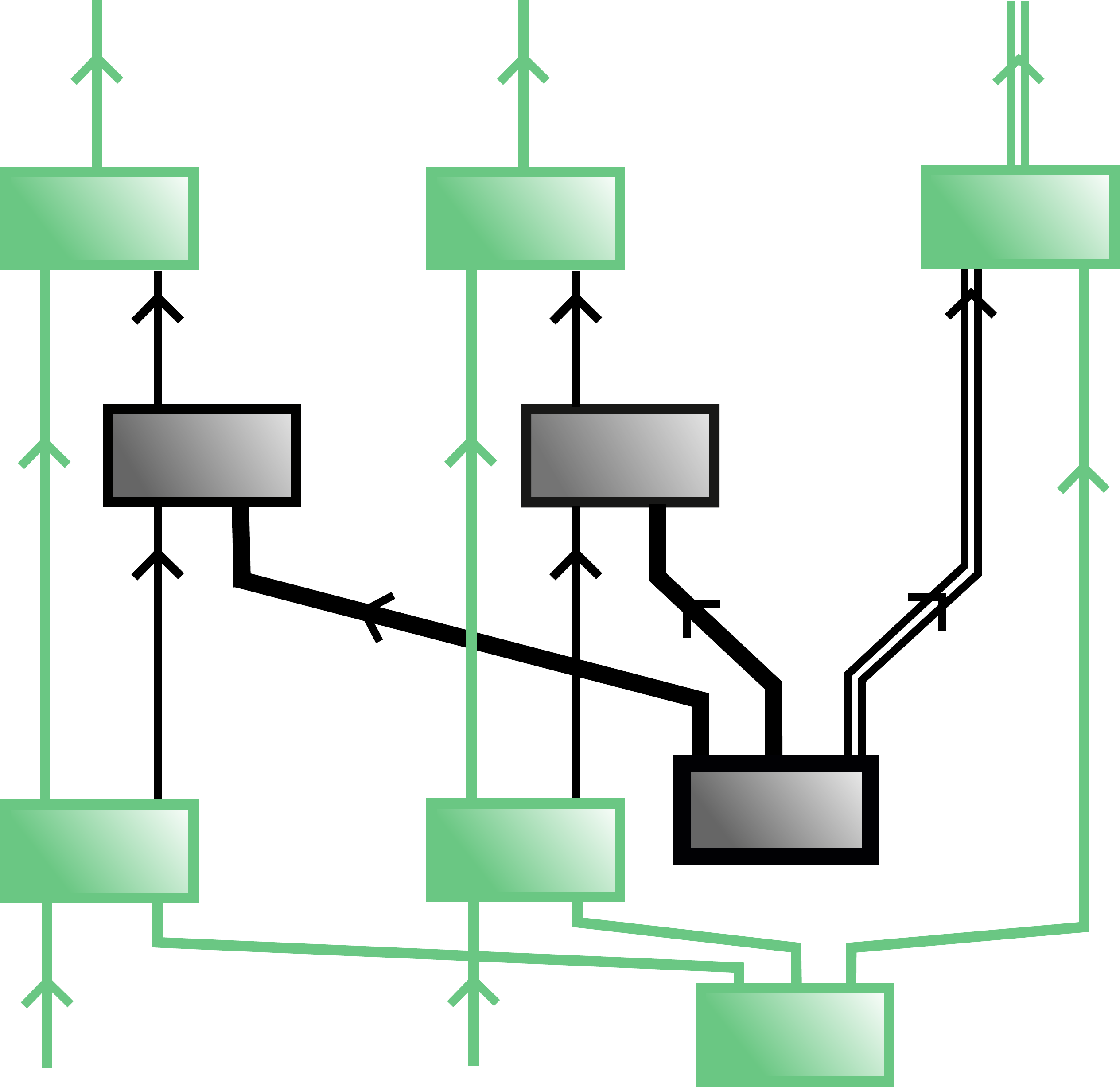}}
\put(-100,71){$x_1$}
\put(-33,71){$x_2$}
\put(-100,118){$a_1$}
\put(-33,118){$a_2$}
\put(-64,90){$c_2$}
\put(-134,90){$c_1$}
\put(-5,5){$\lambda$}
\put(32,118){$\rho_{a_1a_2|x_1x_2}$}
\put(-117,5){$x_1^{\prime}$}
\put(-47,5){$x_2^{\prime}$}
\put(-37,160){$a_2^{\prime}$}
\put(-107,160){$a_1^{\prime}$}
\put(40,160){$\rho_{a_1^{\prime}a_2^{\prime}|x_1^{\prime}x_2^{\prime}}$}
\put(-47,-30){$(b)$}
}

  \end{center}
  \caption{Depiction of a multipartite EPR scenario for $k=2$ (two Alices, one Bob). Quantum systems are depicted by double lines, classical systems by single lines. A system that may be classical or quantum is depicted by a thick line. (a)  Quantum assemblage: Bob receives a quantum system, and the Alices receive systems that may be quantum or classical. Each Alice performs measurements labeled by classical variables $x_1$ and $x_2$, and obtains classical measurement outputs $a_1$ and $a_2$. (b) The most general LOSR operation on a quantum assemblage in a multipartite EPR scenario.}
\end{figure}

In a quantum common-cause multipartite EPR scenario, then, Bob and the Alices share a quantum system, and for each $i \in \{1,\ldots,k\}$, the $i$-th Alice is allowed to perform on her share of the system generalised measurements represented by the set $\{\{M_{a_i|x_i}\}_{a_i}\}_{x_i}$. Then, Bob's unnormalised states can be realised as $\sigma_{a_1...a_k|x_1...x_k}=\text{tr}_{A_1...A_k} \{ (M_{a_1|x_1} \otimes ... \otimes M_{a_k|x_k} \otimes \id) \rho\}$, where $\rho$ is the quantum state of the shared system. These quantumly-realised assemblages are formalised as follows:
\begin{defn}\textbf{Multipartite quantum assemblage.}\\
An assemblage $\As_{\A_1\ldots\A_k|\X_1\ldots\X_k} $ has a quantum realisation iff there exist Hilbert spaces $\cH_{A_i} $ with $i \in \{1,\ldots,k\}$, a state $\rho$ in $\mathcal{H}_{A_1} \otimes \ldots \otimes \cH_{A_k} \otimes \mathcal{H}_B$, and a POVM $\{M_{a_i|x_i}\}_{a_i \in \A}$ on $\mathcal{H}_{A_i}$ for each $x_i \in \X$ and $i \in \{1, \ldots, k\}$, such that 
\begin{align}\label{eq:quantum-multi}
\sigma_{a_1...a_k|x_1...x_k}=\text{tr}_{A_1...A_k} \{ (M_{a_1|x_1} \otimes ... \otimes M_{a_k|x_k} \otimes \id) \rho\}
\end{align}
for all $a_1,\ldots,a_k \in \A$ and $x_1, \ldots, x_k \in \X$.
\label{def:Q-multipartite}
\end{defn}

In analogy to the bipartite case, a multipartite assemblage is LOSR-free if the parties (Alices and Bob) can generate it using local operations (classical and quantum) and classical shared randomness. The elements of an LOSR-free assemblage can be decomposed as $\sigma_{a_1\ldots a_k|x_1\ldots x_k}=\sum_{\lambda} p(\lambda) p^{1}(a_1|x_1,\lambda) \ldots p^{k}(a_k|x_k,\lambda) \, \sigma_{\lambda}$, where $p^{j}(a_j|x_j,\lambda)$ is a conditional probability distribution for the $j$-th Alice for every $\lambda$, and $\sigma_{\lambda}$ are unnormalised quantum states which, similarly to the bipartite case, satisfy $\tr{\sum_\lambda \sigma_\lambda}{} = 1$. 
Notice that this definition of a free multipartite assemblage coincides with the definition of an ``unsteerable multipartite assemblage'' in Refs.~\cite{cavalcanti2011unified,cavalcanti2015detection,cavalcanti2016quantum}, in the sense of ``fully unsteerable''\footnote{This is in analogy to a fully-separable quantum state \cite{LOCCentang}.}. Therefore our LOSR underpinning of the resource-theoretic understanding of assemblages coincides with what people in the literature understand as a ``totally-useless assemblage". Notice, however, that unlike in Ref.~\cite{exposure}, we do not render unsteerable assemblages where the correlations observed among the Alices' measurement outputs are non-classical. In Section~\ref{sec:relwork-multi}, we review different definitions of a multipartite free assemblage and compare them to our approach.  

\subsection{LOSR transformations between multipartite assemblages}

The most general LOSR transformation of a multipartite assemblage consists of a comb for each Alice (which locally pre- and post-processes the variables on each Alice's wing), and a CPTP  map (which post-processes Bob's quantum system), with all these actions being coordinated by a classical random variable $\lambda$. Formally, a generic LOSR operation is illustrated in Fig.~\ref{fig:multipartite-scenario-LOSR}, and transforms an assemblage $\As_{\A_1\ldots \A_k|\X_1\ldots \X_k}$  into a new one as follows:
\begin{align}\label{eq:LOSRtrans-multi} 
\sigma'_{a'_1\ldots a'_k|x'_1\ldots x'_k} = \sum_{\substack{\lambda \\ c_1 \ldots c_k}} \sum_{\substack{a_1\ldots a_k \\ x_1\ldots x_k }} \quad p(c_1,x_1|&x_1^{\prime}, \lambda)  \ldots p(c_k,x_k|x_k^{\prime}, \lambda)  
\quad p(a_1^{\prime}|a_1,c_1) \ldots p(a_k^{\prime}|a_k,c_k) \nonumber \\
& p(\lambda) \quad p(a_1, \ldots, a_k|x_1, \ldots, x_k) \quad
\mathcal{E}_{\lambda}(\rho_{a_1\ldots a_k|x_1\ldots x_k})\,,
\end{align}
where, similarly to the bipartite case,  
\begin{compactitem}
\item $p(c_i,x_i|x_i^{\prime},\lambda) $ encodes the classical pre-processing of the $i$-th Alice's input $x_i$ as a function of $x_i^\prime$ and the shared classical randomness $\lambda$. Here, $c_i$ denotes the variable to be transmitted through the $i$-th Alice's classical side channel towards her post-processing stage. 
\item $p(a_i^{\prime}|a_i,c_i) $ encodes the classical post-processing of the $i$-th Alice's output $a_i$, as a function of the classical information $c_i$ kept from the pre-processing stage. The output of the process is the $i$-th Alice's new outcome $a_i^\prime$. 
\item $ \mathcal{E}_{\lambda}[\cdot]$ is the CPTP map corresponding to Bob's local post-processing of his quantum system, as a function of the shared classical randomness $\lambda$.
\end{compactitem}

Now recall that, similarly to the bipartite case, an LOSR-free assemblage is one that can be created from local operations and shared randomness. Hence, classical assemblages are all and only the ones that can be generated through the free operations of choice, consistent with the unifying assessment of `free of cost' that this resource-theoretic underpinning brings. Moreover, if $\As_{\A_1\ldots \A_k|\X_1\ldots \X_k}$ is free, $\As'_{\A_1^{\prime}\ldots \A_k^{\prime}|\X_1^{\prime}\ldots \X_k^{\prime}}$ is free as well, hence the set of free multipartite assemblages (our free resources) is closed under LOSR operations, as it should be.

A final remark pertains to a particular way to express a generic LOSR transformation, similarly to the discussion in the bipartite scenario. Let us represent each Alice's local comb with a single probability distribution $p(a_{i}^{\prime}, x_i|a_i,x_{i}^{\prime}, \lambda)$. By Fine's argument \cite{FinePRL} and discussion in Ref.~\cite{cowpie}, such local combs can be decomposed as a convex combination of deterministic combs as follows
\begin{align}
    p(a_{i}^{\prime}, x_i|a_i,x_{i}^{\prime}, \lambda) = \sum_{\tilde{\lambda}} p(\tilde{\lambda}|\lambda) D(x_i|x_i^{\prime}, \tilde{\lambda})\,D(a_i^{\prime}|a_i,x_i^{\prime}, \tilde{\lambda}).
\end{align}
Here, each deterministic probability distribution assigns a fixed outcome $a_i^{\prime}$ (resp.~$x_i$) for each possible choice of $a_i$, $x_i^{\prime}$, and $\tilde{\lambda}$ (resp.~$x_i^{\prime}$ and $\tilde{\lambda}$). Let us now use this observation to rewrite Eq.~\eqref{eq:LOSRtrans-multi}; for clarity in the presentation let us focus on the tripartite case (two Alices, one Bob). Define the CPTNI map $\tilde{\mathcal{E}}_{\tilde{\lambda}}$ as:
\begin{align}
\tilde{\mathcal{E}}_{\tilde{\lambda}}(\sigma_{a_1 a_2|x_1 x_2})= \sum_{\lambda} p(\lambda) p(\tilde{\lambda}|\lambda) \mathcal{E}_{\lambda}(p(a_1 a_2|x_1 x_2)\rho_{a_1 a_2|x_1 x_2}).
\end{align}
A generic LOSR operation transforms an assemblage $\As_{\A_1 \A_2|\X_1 \X_2}$  into a new one as follows:
\begin{align}\label{eq:LOSRtrans-multi2} 
\sigma'_{a'_1 a'_2|x'_1 x'_2} = \sum_{\tilde{\lambda}} \sum_{\substack{a_1 a_2 \\ x_1 x_2 }} \, D(x_1|x_1^{\prime}, \tilde{\lambda})\,D(a_1^{\prime}|a_1,x_1^{\prime}, \tilde{\lambda})\, D(x_2|x_2^{\prime}, \tilde{\lambda})\,D(a_2^{\prime}|a_2,x_2^{\prime}, \tilde{\lambda})\,
\tilde{\mathcal{E}}_{\tilde{\lambda}}(\sigma_{a_1 a_2|x_1 x_2}).
\end{align}
Eq.~\eqref{eq:LOSRtrans-multi2} provides the simplified characterisation of a generic multipartite LOSR transformation that we will use throughout.

\subsection{Resource conversion as a semidefinite test}\label{sec:sdp-multi}

For clarity in the presentation, we will still focus on the specific multipartite scenario, illustrated in Fig.~\ref{fig:multipartite-scenario}. Our method, however, extends to scenarios with an arbitrary number of Alices. Given two assemblages generated in this setup, $\As_{\A_1 \A_2|\X_1 \X_2}$ and $\As'_{\A_1^{\prime}\A_2^{\prime}|\X_1^{\prime}\X_2^{\prime}}$, testing whether $\As_{\A_1 \A_2|\X_1 \X_2}$ can be converted into $\As'_{\A_1^{\prime}\A_2^{\prime}|\X_1^{\prime}\X_2^{\prime}}$ under LOSR operations amounts to checking if $\As'_{\A_1^{\prime}\A_2^{\prime}|\X_1^{\prime}\X_2^{\prime}}$ admits a decomposition as per Eq.~\eqref{eq:LOSRtrans-multi2} for the case of two Alices. 
Similarly to the bipartite scenario, the possibility of the conversion can be decided with a single instance of an SDP.

First, notice that the map $\tilde{\mathcal{E}}_{\lambda}(\sigma_{a_1a_2|x_1x_2})$ can be represented in the operator form in terms of its Choi state $W_\lambda$ as follows: 
\begin{align}
\tilde{\mathcal{E}}_{\lambda}(\sigma_{a_1a_2|x_1x_2})= d_B\, \tr{W_{\lambda}\,(\,\id_{B^{\prime}} \otimes \sigma_{a_1a_2|x_1x_2}^{T})}{B}\,,
\end{align}
where $d_B$ is the dimension of Bob's Hilbert space. Therefore, for $\As'_{\A_1^{\prime}\A_2^{\prime}|\X_1^{\prime}\X_2^{\prime}}$ to decompose as in Eq.~\eqref{eq:LOSRtrans-multi2}, each $\sigma'_{a_1^{\prime}a_2^{\prime}|x_1^{\prime}x_2^{\prime}}$ must admit the following decomposition:
\begin{align}\label{eq:sigSDP-multi}
    \sigma'_{a_1^{\prime}a_2^{\prime}|x_1^{\prime}x_2^{\prime}}= \sum_{\substack{\lambda \\ c_1,c_2}}\sum_{\substack{a_1a_2\\x_1x_2}} D(x_1|x_1^{\prime}, {\lambda})\,D(a_1^{\prime}|a_1,x_1^{\prime}, {\lambda})\, D(x_2|x_2^{\prime}, {\lambda})\,&D(a_2^{\prime}|a_2,x_2^{\prime}, {\lambda})\, \nonumber \\
   &d_B\, \tr{W_{\lambda}\,(\,\id_{B^{\prime}} \otimes \sigma_{a_1a_2|x_1x_2}^{T})}{B}\,.
\end{align}
This relation between elements $\sigma'_{a_1^{\prime}a_2^{\prime}|x_1^{\prime}x_2^{\prime}}$ and $\sigma_{a_1a_2|x_1x_2}$ can be tested via an SDP. The SDP that tests whether $\As_{\A_1\A_2|\X_1\X_2}$ can be converted into $\As'_{\A_1^{\prime}\A_2^{\prime}|\X_1^{\prime}\X_2^{\prime}}$ under LOSR operations reads as follows.
\begin{sdp}\label{SDP-multi}\textbf{$\As_{\A_1\A_2|\X_1\X_2}\, \overset{\text{LOSR}}{\longrightarrow} \, \As'_{\A_1^{\prime}\A_2^{\prime}|\X_1^{\prime}\X_2^{\prime}}$.}\\

The assemblage $\As_{\A_1\A_2|\X_1\X_2}$ can be converted into the assemblage $\As'_{\A_1^{\prime}\A_2^{\prime}|\X_1^{\prime}\X_2^{\prime}}$ under LOSR operations if and only if the following SDP is feasible
\begin{align}
\begin{split}
\textrm{given} \;\;\;& \{ \{\sigma_{a_1a_2|x_1x_2}\}_{a_1,a_2}\}_{x_1,x_2}\,,\; \{\{ \sigma'_{a_1^{\prime}a_2^{\prime}|x_1^{\prime}x_2^{\prime}}\}_{a_1^{\prime},a_2^{\prime}}\}_{x_1^{\prime},x_2^{\prime}}\,,\\ 
& \{D(a_1^{\prime}|a_1,x_1^{\prime},\lambda)\}_{\lambda,a_1^{\prime},a_1,x_1^{\prime}}\,,\, D(x_1|x_1^{\prime},\lambda)\}_{\lambda,x_1,x_1^{\prime}} \,, \\
& \{D(a_2^{\prime}|a_2,x_2^{\prime},\lambda)\}_{\lambda,a_2^{\prime},a_2,x_2^{\prime}}\,,\, D(x_2|x_2^{\prime},\lambda)\}_{\lambda,x_2,x_2^{\prime}}  \\
    \textrm{find} \;\;\;& \{(W_{\lambda})_{BB'}\}_{\lambda}  \\
    \textrm{s.t.} \;\;\;& \begin{cases} W_{\lambda} \geq 0\,,\\
      \tr{W_{\lambda}}{B^{\prime}} \propto \frac{1}{d}\,\id_B \;\;\;\; \forall \lambda\,, \\
      \sum_{\lambda} \tr{W_{\lambda}}{B^{\prime}} = \frac{1}{d}\, \id_B\,,\\
          \sigma'_{a_1^{\prime}a_2^{\prime}|x_1^{\prime}x_2^{\prime}}= \sum_{\lambda}\sum_{\substack{a_1a_2\\x_1x_2}} & D(a_1^{\prime}|a_1,x_1^{\prime},\lambda)\, D(x_1|x_1^{\prime},\lambda) \, D(a_2^{\prime}|a_2,x_2^{\prime},\lambda) \\
   &D(x_2|x_2^{\prime},\lambda) \, d_B\, \tr{W_{\lambda}\,(\,\id_{B^{\prime}} \otimes \sigma_{a_1a_2|x_1x_2}^{T})}{B}\,.
     \end{cases}
    \end{split}
\end{align}
\end{sdp}
Similarly to SDP \ref{SDP-bipartite}, SDP \ref{SDP-multi} is a feasibility problem. 

\subsection{Properties of the pre-order}

The properties of the pre-order of multipartite assemblages can be studied numerically, with the SDP \ref{SDP-multi}, or analytically, with EPR LOSR monotones. The SDP \ref{SDP-multi} can be used to test conversions between multipartite assemblages, and a plot similar to that in Fig.~\ref{fig:SDP-bipartite} can be made to illustrate the pre-order of any multipartite family. Hence, we will not repeat this analysis for the multipartite case. In this section, we study the pre-order analytically, but the results can be easily verified with SDP \ref{SDP-multi}. Although our methods apply to general assemblages, here we focus on quantumly-realisable assemblages, and study the properties of the pre-order for a particular family of resources therein.

Consider an EPR scenario with $N=k+1$ parties, where $\A_i=\X_i=\{0,1\}$ for $i\in\{1,...,N-1\}$. Assume all parties share the state $\rho_N^{\theta}$ defined as
\begin{align}\label{eq:Nass}
\rho_N^{\theta}&=\ket{GHZ_N^{\theta}}\bra{GHZ_N^{\theta}} \,, \\
\text{with} \quad \ket{GHZ_N^{\theta}}&=\cos{\theta} \ket{0}^{\otimes N} + \sin{\theta} \ket{1}^{\otimes N} \,, \nonumber 
\end{align}
and the measurements that Alices perform upon an input $x_i \in \X$ are given by
\begin{align}\label{eq:Nmeas}
\widetilde{M}_{a_i|0} &= \frac{\id + (-1)^{a_i} \sigma_z}{2}\,,\quad \widetilde{M}_{a_i|1} = \frac{\id + (-1)^{a_i} \sigma_x}{2}\,. 
\end{align}
Let us define a family of assemblages $\FA^{(N)}$ as: 
\begin{align}\label{eq:thefamN}
\FA^{(N)} &= \left\{ \As^\theta_{\A_1\ldots \A_{N-1}|\X_1\ldots \X_{N-1}} \, \Big\vert \, \theta \in \left(0,\sfrac{\pi}{4}\right] \right\}\,,\\ \nonumber
\text{where} \quad\As^\theta_{\A_1\ldots \A_{N-1}|\X_1\ldots \X_{N-1}} &= \left\{\left\{\sigma^\theta_{a_1...a_{N-1}|x_1...x_{_{N-1}}}\right\}_{a_i\in\A_i}\right\}_{ x_i \in \X_i} \,,\\
\text{with} \quad \sigma^\theta_{a_1...a_{N-1}|x_1...x_{_{N-1}}} &= \tr{\widetilde{M}_{a_1|x_1} \otimes ... \otimes \widetilde{M}_{a_{N-1}|x_{N-1}} \otimes \id \, \rho_N^{\theta}}{A_1...A_{N-1}} \,. \nonumber
\end{align}

For simplicity in the notation, we here denote $\As^\theta_N := \As^\theta_{\A_1\ldots \A_{N-1}|\X_1\ldots \X_{N-1}}$ and $\As_N := \As_{\A_1\ldots \A_{N-1}|\X_1\ldots \X_{N-1}}$.

This family $\FA^{(N)}$, for fixed $N$, has an infinite number of elements. We claim that the elements of $\FA^{(N)} $ are unordered as per the LOSR resource theory of assemblages. To show this, we define a set of EPR monotones $\{S_{\eta}^{(N)}\}_\eta$ using the Bell inequalities derived in Ref.~\cite[Eq.~(13)]{multiGHZ}. This procedure is presented in the Appendix~\ref{app:steeringin-multi}, and the EPR functional is given in Definition \ref{def:NSFT-multi}. For each value of $\eta$, we construct the monotone $\M_{\eta}^{(N)}$ from the EPR functional $S_{\eta}^{(N)}$ following a yield-based construction:

\begin{defn}\label{def:MetaN} The EPR monotone $\M_{\eta}^{(N)}$, for $\eta \in (0,\sfrac{\pi}{4}]$, is defined as
\begin{align}
\M_{\eta}^{(N)}[\As_N] := \max_{\widetilde{\As}_N} \{ S_{\eta}^{(N)}[\widetilde{\As}_N]: \As_N \xrightarrow{\mathrm{LOSR}} \widetilde{\As}_N \}.
\end{align}
\end{defn}
In Appendix~\ref{app:steeringin-multi}, we show that the monotones $\M_{\eta}^{(N)}[\As_N]$ satisfy properties analogous to these given in Eqs.~\eqref{eq:property1} and \eqref{eq:property2}, and Theorems \ref{thm:4} and \ref{thm:M-bipartite}, for the monotones $\M_{\eta}[\As]$. It follows that:
\begin{cor}\label{cor:multi}
The infinite family of EPR monotones $\FM^{(N)} = \{\M_\eta^{(N)} \, \vert \, \eta \in (0,\sfrac{\pi}{4}] \}$ certifies that the infinite family of assemblages $\FA^{(N)}$ is composed of pairwise unordered resources. 
\end{cor}

This result shows how, when LOSR are considered to be free operations in a resource theory, methods used in the bipartite scenarios can be leveraged to the multipartite case.

\section{Related Work}\label{sec:relwork}

Although the causal approach we take in this paper is not the standard description of the EPR scenario, the objects that we study, i.e., assemblages, have been widely investigated for their information-theoretic and  foundational relevance. In particular, their role in one-side device-independent quantum key distribution (1S-DI-QKD) protocols motivated a formulation of the resource theory of `steering', where the set of free operations is deemed to be stochastic local operations assisted by one-way classical communication (S-1W-LOCC) \cite{gallego2015resource}. This set of operations -- just like LOSR -- is a valid set of free operations in a resource theory of assemblages, as it maps free resources to free resources. However, the conceptual underpinnings of our resource theory of nonclassicality of assemblages and of the resource theory of `steering' under S-1W-LOCC operations are very different. In this section, we discuss some conceptual and technical differences between these two resource theories. Moreover, we show that the choice of free operations in a resource theory has significant consequences for the definition of a free multipartite assemblage.

\subsection{S-1W-LOCC as the set of free operations}

\subsubsection{Formalisation of S-1W-LOCC operations}

To start with, let us recall the formal definition of a S-1W-LOCC operation \cite{gallego2015resource} which is illustrated in Fig.~\ref{fig:1W-LOCC}. Consider a bipartite EPR scenario. First, Bob performs an instrument on his subsystem that produces both a classical system represented by the variable $\omega$, and a quantum system. This action corresponds to a completely-positive, trace-non-increasing map $\mathcal{E}_{\omega}$ on Bob's quantum system. Next, he communicates a classical message that depends on $\omega$ to Alice; without loss of generality, we can take the communicated message to be $\omega$ itself. In the second laboratory, Alice at first generates the classical input variable $x'$. Once she has received $\omega$, she generates the new input variable $x$ which can depend on both $x'$ and $\omega$. Then, she performs the measurement labelled by $x$ on her share of the system; she will then obtain a measurement outcome $a$. Finally, Alice classically processes the variables $x'$, $\omega$ and $a$, to produce the classical output $a'$ of her measurement process. In the S-1W-LOCC operations, the instrument does not necessary need to be complete, i.e., if the probability for each outcome $\omega$ is $P_{\omega}$, then the only guarantee on the sum of these probabilities is $0<\sum_{\omega}P_{\omega}\leq 1$. It means that the total S-1W-LOCC operation does not need to happen with certainty, but with some non-zero probability. It was shown in Ref.~\cite{gallego2015resource} that these bipartite S-1W-LOCC operations map free assemblages to free assemblages in bipartite EPR scenarios; hence, they are formally a valid set of free operations in a resource theory of assemblages therein. 

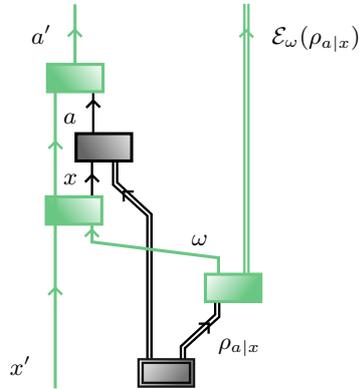
\begin{figure}[h!]
  \begin{center}
  \subcaptionbox
{\put(-60,0){\includegraphics[width=0.18\textwidth]{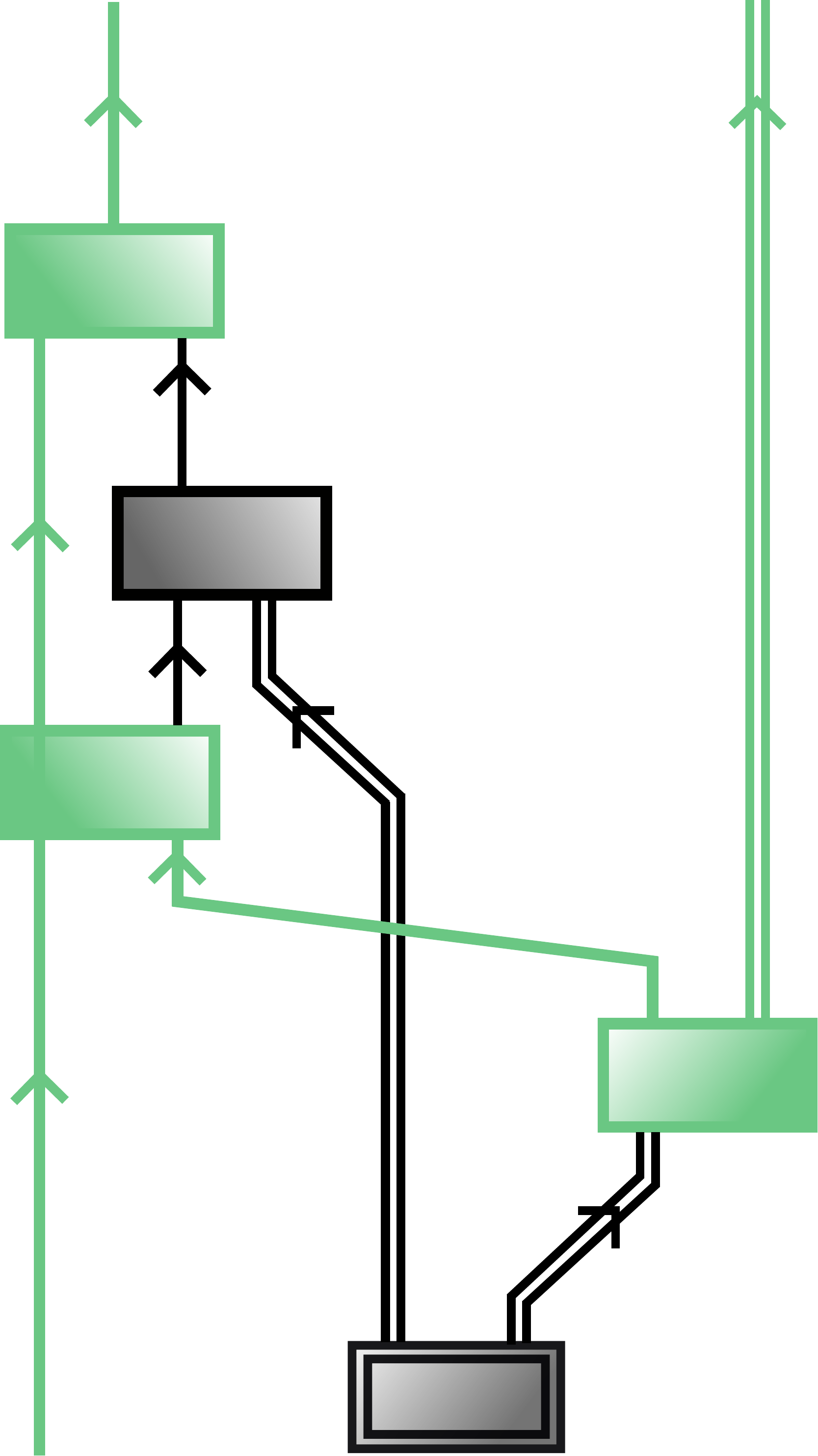}}
\put(-5,54){$\omega$}
\put(-53,77){$x$}
\put(-73,5){$x'$}
\put(-53,100){$a$}
\put(-65,130){$a'$}
\put(25,130){$\mathcal{E}_{\omega} (\rho_{a|x})$}
\put(5,15){$\rho_{a|x}$}
}

  \end{center}
  \caption{The most general 1W-LOCC operation (in green) applied to an EPR assemblage (in black).}
  \label{fig:1W-LOCC}
\end{figure}

 \subsubsection{Comparing S-1W-LOCC and LOSR}\label{sec:relwork-comparison}

The resource theory of `steering' under S-1W-LOCC and our resource theory of nonclassicality of assemblages are both formally valid resource theories that quantify nonclassicality in the bipartite EPR scenario. In this subsection, we will argue that the causal view on assemblages has several advantages over the S-1W-LOCC approach. 

The set of free operations in any resource theory should be motivated by natural physical constraints in the scenario of interest. The set of S-1W-LOCC operations is hence a meaningful choice in scenarios wherein one-way classical communication is possible. Indeed, the choice of S-1W-LOCC as the set of free operations was motivated by the fact that it is the most general set that does not compromise the security of 1S-DI-QKD protocols~\cite{gallego2015resource}. 
 In contrast, the motivation behind our LOSR approach is fundamentally different. Rather than using a methodology which singles out one particular 
 task for which assemblages are known to be useful, we take a principled approach that applies to any scenario with a common-cause causal structure. In particular, we follow a general construction that can be applied to study nonclassicality in arbitrary causal networks, as first introduced in Ref.~\cite{cowpie} (see in particular Appendix~A3 therein). 
Indeed, our resource theory of EPR assemblages is simply a type-specific instance of the type-independent resource theory introduced in Ref.~\cite{schmid2020type}. This is the most significant conceptual advantage of our LOSR approach -- it unifies the study of nonclassicality of assemblages with the study of nonclassicality of arbitrary processes in common-cause (e.g., Bell) scenarios. This unified view leads to a deeper understanding of nonclassical correlations, as well as to new technical tools for studying them~\cite{schmid2020standard,schmid2020type,rosset2020type}. 

It is also worth noting that, at a technical level, LOSR conversions of assemblages are much easier to characterize and study than the S-1W-LOCC conversions. 

The fact that we have taken a principled approach to constructing our resource theory also implies a final advantage: our framework extends naturally and uniquely to more general scenarios, including multipartite EPR scenarios, as well as EPR scenarios where the processes about which Alice is learning are dynamic (e.g., general channels or Bob-with-input processes\footnote{Bob-with-input processes are those in which Bob can locally influence the state preparation
of his system~\cite{sainz2020bipartite}.}). It is not clear how prior approaches generalize to these cases; indeed, in past work there is no consensus even on the simplest question of how to define free assemblages in multipartite EPR scenarios (we discuss this point further in the following sections).

 Finally, for completeness, we note that the LOSR and S-1W-LOCC resource theories do in fact lead to different pre-orders, and hence give different answers to questions about the relative value of assemblages.

\begin{prop}\label{prop:pre-order}
The pre-order of assemblages under LOSR operations is different than the pre-order of assemblages under S-1W-LOCC operations.
\end{prop}

The proof of this proposition is given in Appendix~\ref{app:proof-prop-preorder}. The main idea is to find a conversion that is possible with S-1W-LOCC operations and is impossible with LOSR operations. In the proof we focus on the family $\As^{\theta,p}_{\A|\X}$ defined in Eq.~\eqref{eq:thefam}. We show that for S-1W-LOCC operations, the assemblage $\As^{\pi/4,1}_{\A|\X}$ is above all others in the family; whereas in LOSR, all assemblages $\As^{\theta,1}_{\A|\X}$ are incomparable for $\theta \in \left(0,\sfrac{\pi}{4}\right]$. In Appendix~\ref{app:proof-prop-preorder} we provide an S-1W-LOCC map that maps $\As^{\pi/4,1}_{\A|\X}$ to $\As^{\theta,1}_{\A|\X}$ for $\theta \in \left(0,\sfrac{\pi}{4}\right)$. This result is analogous to the difference in the pre-order in the resource theories of LOCC-entanglement and LOSR-entanglement \cite{schmid2020standard}\footnote{There are two different resource theories of entanglement established in the literature. If one allows local operations and classical communication (LOCC) as the set of free operations, it is well known that a maximally entangled state can be converted into any partially entangled pure state of the same Schmidt rank~\cite{LOCCentang}. However, if one takes LOSR to be the set of free operations, such states are incomparable~\cite{schmid2020standard}. Here, we use the terms LOCC-entanglement and LOSR-entanglement to distinguish these two resource theories  \cite{LOCCentang,schmid2020standard}. 
}. 

In proving Proposition~\ref{prop:pre-order},  we note a deficiency in Theorem~5 in Ref.~\cite{gallego2015resource}. The  theorem  claims  that  there  is  no universal assemblage in a particular setting from which we obtain any other assemblage in the same setting with S-1W-LOCC maps.  The proof applies if one assumes that the universal assemblage is not $\As^{\pi/4,1}_{\A|\X}$, thus limiting the applicability of the theorem. Therefore, the proof given in Ref.~\cite{gallego2015resource} is incomplete.

\subsection{Multipartite free assemblages}\label{sec:relwork-multi}

Historically, a bipartite assemblage would be considered nonfree or nonclassical\footnote{In the literature, nonfree assemblages are referred to as `steerable' assemblages.} if Alice and Bob need to share a quantum system prepared in an entangled state in order to produce it. If the assemblage could be generated by performing measurements on a system prepared on a separable state, then the assemblage is free. 

{\em Both} for bipartite {\em and} multipartite scenarios, this is exactly the distinction between free and nonfree in our approach: an assemblage is free if and only if it can be generated by measurements on a fully-separable state. Although the definition of a free multipartite assemblage coincides with our approach in some past literature~\cite{cavalcanti2011unified,cavalcanti2015detection,cavalcanti2016quantum}, alternative approaches to studying multipartite EPR scenarios also exist . We now discuss some of the different sets of assemblages that have been proposed as free sets in the multipartite case, and contrast them with our proposal.

For simplicity, we will focus on the case of a tripartite EPR scenario with two Alices and one {Bob, although the three definitions easily generalise to EPR scenarios with more than two Alices.} In this setup, one proposal is given in Ref.~\cite{uola2020quantum}, where an assemblage is called free if it can be decomposed as
\begin{align}\label{eq:uola_unsteer}
    \sigma_{a_1a_2|x_1x_2}=\sum_\lambda p(\lambda) p(a_1,a_2|x_1,x_2,\lambda) \rho_{\lambda},
\end{align}
where $p(\lambda)$ is a normalised probability distribution on $\lambda$, $\rho_\lambda$ is a normalised quantum state for each $\lambda$, and $p(a_1,a_2|x_1,x_2,\lambda)$ is only required to be a valid conditional probability distribution. Indeed, here $p(a_1,a_2|x_1,x_2,\lambda)$ is allowed to even be a signalling correlation between the Alices, for each $\lambda$, as long as this signalling is averaged out in Eq.~\eqref{eq:uola_unsteer} to give rise to a valid quantum assemblage
$\{\{\sigma_{a_1a_2|x_1x_2}\}_{a_1,a_2\in\A}\}_{ x_1,x_2 \in \X}$.
This condition for $\{p(a_1,a_2|x_1,x_2,\lambda)\}$ is a form of fine-tuning~\cite{Wood_2015}, and is conceptually problematic. 
Furthermore, it is clear that any assemblage that is nontrivially fine-tuned in this manner is not consistent with our hypothesis about the causal structure of an EPR scenario. Note also that, by this definition, there exist assemblages that are free and yet for which the correlations shared by the Alices (when one marginalizes over Bob's system) are highly nonclassical-- e.g., a Popescu-Rohrlich box correlations \cite{popescu1994quantum}. As such, this approach is not suited to studying nonclassicality in general, but could only be suitable for studying nonclassicality {\em across the Alices-vs-Bob} partition. A different proposal for the definition of a free assemblage is given in Ref.~\cite{exposure}, where an assemblage is deemed free if it can be expressed as
\begin{align}\label{eq:TO-LHS}
\begin{split}
    \sigma_{a_1a_2|x_1x_2}&=\sum_{\lambda} p(\lambda) p^{A_1 \rightarrow A_2}(a_1,a_2|x_1,x_2,\lambda) \rho_{\lambda} \\
    &=\sum_{\lambda} p^{\prime}(\lambda) p^{A_2 \rightarrow A_1}(a_1,a_2|x_1,x_2,\lambda) \rho^{\prime}_{\lambda},
    \end{split}
\end{align}
where for each value of $\lambda$, the probability distribution $p^{A_1 \rightarrow A_2}(a_1,a_2|x_1,x_2,\lambda)$ is non-signalling from Alice$_2$ to Alice$_1$, and the probability distribution $p^{A_2 \rightarrow A_1}(a_1,a_2|x_1,x_2,\lambda)$ is non-signalling from Alice$_1$ to Alice$_2$. Such a decomposition is called the time-ordered local hidden state model (TO-LHS) of an assemblage. This set has been motivated by the fact that it is the largest set of free assemblages for which an anomalous phenomenon termed `exposure of steering'~\cite{exposure} does not occur (under particular assumptions about the set of free operations, distinct from those we have advocated for), and we elaborate further on this point in Section~\ref{sec:relwork-exposure}. 
 In particular, we note that in any standard resource theory (i.e., in the sense of Ref.~\cite{coecke2016mathematical}), both the free resources and the free transformations on them should be derivable from a single set of free operations.
What this implies here is that it should be the case that the assemblages which admit of TO-LHS models constitute {\em all and only} those assemblages which can be constructively generated using S-1W-LOCC operations. This is, however, not the case, since the set of TO-LHS assemblages includes some post-quantum tripartite assemblages (e.g., where the Alices share a Popescu-Rohrlich box; see the next subsection), and these cannot be generated by a S-1W-LOCC protocol.

These two definitions of free assemblage contrast with our LOSR resource theory of nonclassicality of assemblages, where an assemblage is LOSR-free if it can be written as
\begin{align}\label{eq:unsteer_multi}
    \sigma_{a_1a_2|x_1x_2}=\sum_\lambda p(\lambda) p(a_1|x_1,\lambda) p(a_2|x_2,\lambda) \sigma_{\lambda}.
\end{align}
Classicality in our resource-theoretic approach is captured by the structure of the network, and imposes that the assemblage is free, i.e., classical, if the parties -- Alices and Bob --  are related solely by a shared classical common cause. This definition of a free assemblage has indeed appeared in the literature before in the context of EPR scenarios, and not from a resource-theoretic perspective: in Refs.~\cite{cavalcanti2011unified,cavalcanti2015detection,cavalcanti2016quantum}, a free assemblage is defined as one that arises when all parties share a quantum system prepared in a fully separable state, and the Alices perform local measurements on their subsystems -- this leads to all and only assemblages of the form of Eq.~\eqref{eq:unsteer_multi}.

\bigskip

Finally, it is worth noting that in this discussion we have focused on multipartite EPR scenarios with multiple Alices and one Bob. However, a possible multipartite EPR scenario comprises also multiple Bobs. Such a multipartite scenario is beyond the scope of this work, however, it would be interesting to see which other approaches to defining free assemblages it gives rise to\footnote{In the case of our LOSR resource theory of nonclassical assemblages the situation is straightforward: all the parties share a classical common cause, and each Bob locally prepares a quantum system on a state that depends on the value of such shared classical randomness.}.

\subsection{Exposure of `steering'}\label{sec:relwork-exposure}
 
\textit{Exposure of `steering'} refers to a process in which a resourceful assemblage in an EPR scenario can be created from a free assemblage in a multipartie EPR scenario, by performing a global operation on the Alices. Notice that such an operation may change  the type of scenario under study; e.g., it may transform a tripartite assemblage into a bipartite one. 

A type-changing transformation of this nature is considered in Ref.~\cite{exposure}, where Taddei~\etal ~focus on a tripartite EPR scenario with two Alices and one Bob. As a case study, they take a free tripartite assemblage to be defined as per Eq.~\eqref{eq:uola_unsteer}, when no extra constraints are imposed on $p(a_1,a_2|x_1,x_2,\lambda)$ other than the implicit fine tuning condition. The global operation on the two Alices is a wiring operation -- the measurement outcome of one Alice is taken as the choice of measurement for the other Alice. Taddei~\etal ~argue that this global operation should not increase the Alices' capability to remotely 
`steer' Bob's system, since no global action is performed across the Bob-vs-Alices divide, i.e., this global operation should be a free operation (at least relative to the Bob-vs-Alices divide). Despite this intuition, they find that a wiring between the Alices in a tripartite free assemblage may give rise to a new bipartite assemblage that is nonfree. 

This then raises important conceptual questions, since a well-defined 
\textit{compositional} 
resource theory should not allow one to create nonfree resources by using free operations on free resources, \textit{even if} the free transformation is type-changing. 
The set of free resources in any resource theory must be closed under free operations~\cite{coecke2016mathematical}. 
The spectre of `steering' exposure emphasizes the importance of choosing one's set of free resources and free operations wisely and consistently.

As Taddei \etal ~observe, then, `steering' exposure implies an inconsistency taking Eq.~\eqref{eq:uola_unsteer} to define the set of free resources (together with taking wirings-and-S-1W-LOCC as free operations).
They then derive the largest set of assemblages that do not allow for `steering' exposure under these free operations, and show this set to be those assemblages that admit a TO-LHS model.
Hence, they argue that the set of free assemblages should be taken to be defined by Eq.~\eqref{eq:TO-LHS}.  

One awkwardness of this proposal is that the set of TO-LHS assemblages includes assemblages that are not quantumly-realizable; i.e., that cannot be produced by Bob and the Alices sharing a multipartite quantum system and performing local measurements on it. An example of this arises when $p^{A_1 \rightarrow A_2}(a_1,a_2|x_1,x_2,\lambda) = p^{A_2 \rightarrow A_1}(a_1,a_2|x_1,x_2,\lambda)$ is equal to `a Popescu-Rohrlich box' \cite{popescu1994quantum} for all $\lambda$. This sort of definition, then, is clearly not suitable for studying generic nonclassicality in multipartite EPR scenarios, but rather could {\em only} be suitable as a way of studying {\em nonclassicality across the Alices-vs-Bob partition}. 

Another problem with Taddei \etal's approach is that it comes at the expense of introducing an arbitrary divide between free resources and free processes. That is, as argued in the previous section, the set of free resources defined by Eq.~\eqref{eq:TO-LHS} does not constitute all and only those that can be constructively built out of S-1W-LOCC operations, as the standard approach to resource theories requires. Indeed, any TO-LHS assemblage for which the marginal correlations shared between the Alices is post-quantum will be a counterexample, since the classical communication from Bob to the Alices is not sufficient to set up such correlations.

Note that Taddei \etal's approach is only {\em one} way to avoid exposure phenomena. A second way to avoid it is to consider a different set of free operations. Indeed, in our approach, neither the free operations {\em nor} the free resources are taken to be the sets that Taddei \etal~considered. Furthermore, our approach does not allow for any exposure phenomena, nor does it face the issues we just outlined which arise in Taddei \etal's TO-LHS approach. Note first that in our approach, wirings are not free LOSR operations, since our approach does not only aim to characterize nonclassicality between the Alices-vs-Bob partition, but rather {\em any} nonclassicality of common cause, shared between {\em any} of the parties. As such, wirings among black-box parties can of course increase the nonclassicality of a given common-cause process.
In our approach, then, the question of exposure is whether nonfree transformations that act globally on the Alices---but do not act on Bob---are capable of transforming free resources to resources which require nonclassical common causes {\em shared across the Alices-vs-Bob divide}. 
But such exposure phenomena are not possible in our approach. Indeed, the set of LOSR-free assemblages is strictly contained within the set of TO-LHS assemblages, and hence (by Taddei \etal's result) they cannot yield exposure. 

For completeness, we further note that in our resource theory, it is indeed the case that the free resources are all and only those that can be constructed out of free (i.e., LOSR) operations, so that the consistency condition mentioned above is satisfied.

\section{Conclusions and outlook}

In this paper, we developed a resource theory of nonclassicality for assemblages in EPR scenarios, with its free operations defined as local operations and shared randomness. This choice of free processes is motivated by the causal structure of the EPR scenario, together with a viewpoint on EPR scenarios wherein Alice's actions allow her to make \textit{inferences about} the physical state of Bob's system rather than to \textit{remotely steer} (i.e.,~influence) it. 
This is the first resource theory for assemblages in multipartite EPR scenarios. 

We proved that resource conversions in our resource theory can be tested using a single instance of a semidefinite program, in both bipartite and multipartite EPR scenarios. These semidefinite programs -- which we give -- are the first tools that allow one to explore systematically free conversions of assemblages. We also proved that the pre-order of (both bipartite and multipartite) assemblages contain an infinite number of incomparable assemblages. To prove this, we constructed new EPR monotones, which may be useful in their own right for other tasks.  

The causal approach that we have endorsed here brings a new perspective to fundamental questions regarding EPR scenarios. Firstly, we have motivated a principled approach to defining the set of classically-explainable assemblages in multipartite EPR scenarios from a resource-theoretic perspective. 
In addition, our approach (in contrast to the approach of Ref.~\cite{exposure}) ensures that the free set of resources and the free set of operations are consistent, in the sense that the free resources are all and only those that can be constructed out of free operations. Finally, our approach ensures that there are no free type-changing operations that do not preserve the free set (i.e., that there is no `steering exposure'), and sheds light on previous approaches to ensuring this property. 
As far as we know, this is the first resource theory defined for EPR assemblages that satisfies these (in our opinion) important properties, highlighting the value of taking our principled LOSR approach.

The resource theory of nonclassicality of assemblages is an instance of a general construction for studying the nonclassicality of common-cause processes. This principled approach was previously used to study LOSR-entanglement~\cite{schmid2020standard}, nonclassical correlations in Bell scenarios \cite{cowpie}, and indeed arbitrary types of common-cause processes~\cite{schmid2020type,rosset2020type}. Our work continues the development of this program in the specific context of EPR scenarios, and in doing so provides a perspective on EPR `steering' which unifies it with other pertinent notions of nonclassicality. 
A natural continuation of this work would be to study other EPR-like scenarios, wherein the process of Bob about which Alice is learning is not merely a quantum output system, but rather a dynamical process. Indeed, in our forthcoming paper we develop a resource theory for Bob-with-input scenarios and channel EPR scenarios.

Another relevant research avenue would be to extend the enveloping theory to include post-quantum common-cause processes~\cite{sainz2015postquantum,sainz2020bipartite,schmid2021postquantum} (e.g., described by processes in arbitrary generalized probabilistic theories~\cite{hardy2001quantum, barrett2007information,chiribella2009theoretical}). Scenarios such as the multipartite EPR scenario and the Bob-with-Input scenario are of particular importance for these questions because -- unlike the bipartite EPR scenario -- they allow for post-quantum assemblages. In order to quantify the nonclassicality of such assemblages, one would again take the free set to be defined by LOSR operations, and apply all the tools and monotones developed herein. Alternatively, one could study not the nonclassicality, but rather the {\em post-quantumness} of these resources, by taking the free set to be defined by \textit{local operations and shared entanglement} (LOSE operations), as argued in Ref.~\cite{schmid2021postquantum,hoban2018channel}. In this latter resource theory, one would be studying not the possibilities allowed by quantum theory, but rather the limitations imposed by quantum theory.

Given that both entanglement and incompatibility are necessary resources for generating nonclassical assemblages, another interesting line of work would be to explore how our resource theory of EPR assemblages emerges from the resource theories of incompatibility~\cite{buscemi2020complete} and entanglement~\cite{schmid2020standard}. 

Finally, an interesting question pertains to the notion of self-testing of quantum states, which is relevant for the certification of quantum devices in communication and information processing protocols.  Ref.~\cite{schmid2020standard} shows that self-testing of states can indeed be defined in a resource-theoretic framework. It would be interesting to see how self-testing of quantum assemblages (cf.  Refs.~\cite{supic2016, gheorghiu2017,Chen2021robustselftestingof}) may be understood from the tools we developed in our work, such as the resource-conversion SDP tests.

\section*{Acknowledgments}
BZ, DS, and ABS acknowledge support by the Foundation for Polish Science (IRAP project, ICTQT, contract no. 2018/MAB/5, co-financed by EU within Smart Growth Operational Programme). MJH and ABS acknowledge the FQXi large grant ``The Emergence of Agents from Causal Order'' (FQXi FFF Grant number FQXi-RFP-1803B). BZ acknowledges partial support by the National Science Centre, Poland 2021/41/N/ST2/02242. This research was supported by Perimeter Institute for Theoretical Physics. Research at Perimeter Institute is supported in part by the Government of Canada through the Department of Innovation, Science and Economic Development Canada
and by the Province of Ontario through the Ministry of Colleges and Universities. The diagrams within this manuscript were prepared using TikZ and Inkscape.

\bibliography{LOSRbib}
\bibliographystyle{quantum}

\appendix

\section{The bipartite tilted Bell inequalities}\label{app:bipartite-Bell}

Consider a bipartite Bell scenario with two dichotomic measurements per party. 
The Bell inequality studied in Refs.~\cite{bamps2015sum, acin2012randomness} for this scenario, defined in \cite[Eq.~(1)]{bamps2015sum}, reads: 
\begin{align}\label{eq:app-tiltedBell}
I_\alpha = \alpha \av{A_0} + \av{A_0B_0} + \av{A_0B_1} + \av{A_1B_0} - \av{A_1B_1}\,,
\end{align}
where $\alpha \in [0,2]$ is a real parameter. For $\alpha = 0$, $I_0$ correspond to the CHSH inequality. 

The set $\I = \{I_\alpha \, \vert \, \alpha \in [0,2] \}$ defines a family of Bell inequalities for the bipartite Bell scenario, indexed by the parameter $\alpha$. It was proved in Ref.~\cite{acin2012randomness} that the maximum quantum violation of $I_\alpha$ is given by $I_\alpha^{\max} = \sqrt{8+2\alpha^2}$, and the classical bound of $I_\alpha$ is given by $I_\alpha^{\mathrm{C}} = 2 + \alpha$. The value $I_\alpha^{\max}$ is achieved by measuring the observables 
\begin{align}\label{eq:app-observables}
A_0 &= \sigma_z\,,\quad & B_0= \cos(\mu) \, \sigma_z + \sin(\mu) \, \sigma_x\,,\\ \nonumber
A_1 &=\sigma_x  \,,\quad &B_1= \cos(\mu) \, \sigma_z - \sin(\mu) \, \sigma_x \,,
\end{align}
on the quantum state
\begin{align}\label{eq:app-thetastate}
\ket{\theta}=\cos{\theta} \ket{00} + \sin{\theta} \ket{11} \,, 
\end{align}
with the relation between the parameters $\alpha$, $\theta$, and $\mu$ being the following: 
\begin{align}
\alpha = \frac{2}{\sqrt{1+2\tan^2(2\theta) }}\,,\quad \tan(\mu) = \sin(2\theta)\,.
\end{align}
In Ref. \cite[Section III and Appendix A]{bamps2015sum}, it was proved that the Bell inequalities specified by Eq.~\eqref{eq:app-tiltedBell} provide a robust self-test for the reference states given in Eq.~\eqref{eq:app-thetastate}. 

\begin{rem}\label{rem:STarg}
Let $\widetilde{A}_x$, $\widetilde{B}_y$, and $\ket{\widetilde{\psi}}$ be physical observables and a physical quantum state that achieve the value $I_\alpha^{\max}$ for the Bell inequality $I_\alpha$. Then, these are equal to (up to local isometries) the state and observables given in Eqs.~\eqref{eq:app-thetastate} and \eqref{eq:app-observables}. 

\end{rem}

\section{Constructing bipartite EPR inequalities from Bell inequalities}\label{app:steeringin-bipartite}

We now show how to use the tilted Bell inequalities from the set $\I$, defined in Appendix \ref{app:bipartite-Bell}, to construct a family of EPR inequalities~\cite{cavalcanti2009experimental}. Let us begin by transforming a Bell functional $I_\alpha \in \I$ into an EPR functional $S$. A generic EPR functional is defined as: 
\begin{align}
S[\As] = \Tr{ \sum_{a\in \A, \, x\in \X} F_{a,x} \, \sigma_{a|x}}\,.
\end{align}
The Bell expression $I_\alpha$ can be expressed in terms of the elements of $\As$ as follows:
\begin{align}
I_\alpha &= \alpha \av{A_0} + \av{A_0B_0} + \av{A_0B_1} + \av{A_1B_0} - \av{A_1B_1} \nonumber \\
&= \alpha \Tr{(M_{0|0} - M_{1|0})\otimes \id_B \, \rho} + \Tr{(M_{0|0} - M_{1|0})\otimes B_0 \, \rho} + \Tr{(M_{0|0} - M_{1|0})\otimes B_1 \, \rho} \nonumber \\
&\quad + \Tr{(M_{0|1} - M_{1|1})\otimes B_0 \, \rho} - \Tr{(M_{0|1} - M_{1|1})\otimes B_1 \, \rho} \nonumber \\
&= \Tr{\alpha \, \sigma_{0|0} - \alpha \, \sigma_{1|0} + B_0 \, \sigma_{0|0} - B_0 \, \sigma_{1|0} + B_1 \, \sigma_{0|0} - B_1 \, \sigma_{1|0} + B_0 \, \sigma_{0|1} - B_0 \, \sigma_{1|1}  } \nonumber \\
&\quad +\Tr{- B_1 \, \sigma_{0|1} + B_1 \, \sigma_{1|1}} \nonumber
\end{align}
From the above expansion of $I_\alpha$, we then define an EPR functional $S[\As]$ by setting: 
\begin{align}\label{eq:theF}
F_{0,0} = \alpha \id + B_0 + B_1 = - F_{1,0} = \,,\, 
F_{0,1} = B_0 - B_1 = - F_{1,1} \,.
\end{align}
In the expression for  $S[\As]$ given by the operators of Eq.~\eqref{eq:theF} there are plenty of things that still need to be specified: the value of the parameter $\alpha$, as well as the operators $B_0$ and $B_1$. This is the freedom we will leverage to construct a family of EPR inequalities. 

\begin{defn}\label{def:SFT} \textbf{EPR functional $S_\eta[\As]$} \\
We define the EPR functional $S_\eta[\As]$ via the operators of Eq.~\eqref{eq:theF} by taking: 
\begin{align}\label{eq:theStpar}
\alpha &= \frac{2}{\sqrt{1+2\tan^2(2\eta) }}\,, \\ \nonumber
B_0 &= \cos(\mu) \, \sigma_z + \sin(\mu) \, \sigma_x \,,\\ \nonumber
B_1 &= \cos(\mu) \, \sigma_z - \sin(\mu) \, \sigma_x \,, \\ \nonumber
\tan(\mu) &= \sin(2\eta)\,.
\end{align}
\end{defn}
A few remarks regarding the properties of $S_\eta[\As]$ are in order.
\begin{rem}
The maximum quantum value $S_\eta^{\max}$ of $S_\eta[\As]$ is given by 
\begin{align}\label{eq:Smax}
S_\eta^{\max} = 2\sqrt{2} \, \sqrt{1 + \frac{1}{1+2\tan^2(2\eta)}}\,.
\end{align}
\end{rem}
\begin{proof}
First notice that `optimising $S_\eta[\As]$ over quantum assemblages' is equivalent to `optimising $I_\alpha$ over quantum correlations constrained on Bob's observables being those from Eq.~\eqref{eq:theStpar}, and on $\alpha$ and $\eta$ being related as in Eq.~\eqref{eq:theStpar}'. 

Since the maximum quantum violation $I_\alpha^{\max}$ can indeed be achieved by the operators for Bob from Eq.~\eqref{eq:theStpar}, then $S_\eta^{\max} = I_\alpha^{\max}$. Using the relation between $\alpha$ and $\eta$ we obtain  Eq.~\eqref{eq:Smax}. 
\end{proof}

\begin{rem}\label{rem:iff}
The quantum assemblage $\As^{\theta}_{\A|\X}$, defined in Eq.~\eqref{eq:thefam1}, satisfies $S_\eta[\As^{\theta}_{\A|\X}] = S_\eta^{\max}$ if and only if $\eta=\theta$.
\end{rem}
\begin{proof}
The `if' direction is straightforward to prove, by noticing that the state and measurements used to prepare $\As^{\theta}_{\A|\X}$, together with Bob's observables from Eq.~\eqref{eq:theStpar}, produce correlations that achieve $I_\alpha^{\max}$ in the corresponding Bell experiment, where $\alpha$ and $\eta$ are related as in Eq.~\eqref{eq:theStpar}. 

The `only if' direction follows from a proof by contradiction.
Assume that there exists $\As^{\theta}_{\A|\X}$ with $\theta \neq \eta$, with $S_\eta[\As^{\theta}_{\A|\X}] = S_\eta^{\max}$. 
Then,  the state and measurements used to prepare $\As^{\theta}_{\A|\X}$, together with Bob's observables from Eq.~\eqref{eq:theStpar}, become a quantum realisation of correlations that achieve $I_\alpha^{\max}$ for the Bell functional $I_\alpha$, where $\alpha = \frac{2}{\sqrt{1+2\tan^2(2\eta) }}$. 

By Remark \ref{rem:STarg} then one concludes that there exists a local isometry in Alice's lab that takes $\ket{\theta}$ to $\ket{\eta}$. This local isometry, however, can never exist, since $\ket{\theta}$ cannot be transformed into $\ket{\eta}$ by such local operations \cite{schmid2020standard}. Hence, it must be the case that $\theta = \eta$. 
\end{proof}

\begin{rem}\label{rem:thelast}
Let $\As_{\A|\X}$ be a quantum assemblage, and let $\rho$ and $\{\{M_{a|x}\}_{a\in \A}\}_{x\in \X}$ be a quantum state and measurements that realise it. Then, $S_\eta[\As_{\A|\X}] = S_\eta^{\max}$ if and only if: \\
- The observables $A^\prime_0 = M_{0|0} - M_{1|0}$ and $A^\prime_1 = M_{0|1} - M_{1|1}$ are equivalent (up to local isometries) to $A_0$ and $A_1$, respectively, from Eq.~\eqref{eq:app-observables},\\
- The quantum state $\rho$ is equivalent, up to local isometries, to $\ket{\eta}$. 
\end{rem}
\begin{proof}
The `if' direction follows trivially from Remark \ref{rem:iff}. \\
To prove the `only if' direction, notice that $\rho$, $\{A^\prime_j\}_{j=0,1}$, and Bob's observables $\{B_j\}_{j=0,1}$ from Eq.~\eqref{eq:theStpar}, provide a quantum realisation of correlations that achieve the maximum value of the Bell functional $I_\alpha$, with $\alpha$ and $\eta$ related as in Eq.~\eqref{eq:theStpar}. From Remark \ref{rem:STarg} it follows that $\rho$ and $\{A^\prime_j\}_{j=0,1}$ must be equal to (up to local isometries) the state and observables of Eqs.~\eqref{eq:app-thetastate} and \eqref{eq:app-observables}, which proves our claim. 
\end{proof}

\section{The multipartite Bell inequalities}\label{app:multipartite-bell}

Consider $N$-partite Bell scenario with two dichotomic measurements per party. The Bell inequality studied in Ref.~\cite{multiGHZ} for this scenario, defined in Ref.~\cite[Eq.~(13)]{multiGHZ}, reads: 
\begin{align} \label{eq:app-multi-Bell}
\begin{split}
I_\alpha^{(N)} &= (N-1) \av{(B_0+B_1)A_0^{(1)}...A_0^{(N-1)}} 
+ (N-1)  \frac{\cos{2\alpha}}{\sqrt{1-\cos^2{2\alpha}}} (\av{B_0}-\av{B_1}) \\
&+ \frac{1}{\sqrt{1-\cos^2{2\alpha}}} \sum_{i=1}^{N-1} \av{(B_0-B_1)A_1^{(i)}}\,,
\end{split}
\end{align}
where $\alpha \in (0,\pi/4]$. 

The set $\I^{(N)} = \{I_\alpha^{(N)} \, \vert \, \alpha \in (0,\pi/4] \}$ defines, for fixed $N$, a family of Bell inequalities for the multipartite scenario, indexed by the parameter $\alpha$. It was shown in Ref.~\cite{multiGHZ} that the maximum quantum violation of $I_\alpha^{(N)}$ is given by $I_\alpha^{(N) \max} = 2\sqrt{2}(N-1)$, and the classical bound of $I_\alpha^{(N)}$ is given by $I_\alpha^{(N)\,\mathrm{C}} = (N-1) \frac{1-\cos{2\alpha}}{\sqrt{1-\cos{2\alpha}}}$. The value $I_\alpha^{(N) \max}$ is achieved by measuring the observables 
\begin{align}\label{eq:app-multi-observables}
A_0^{(i)} &= \sigma_x\,,\quad & B_0=  \cos(\mu) \, \sigma_z +  \sin(\mu) \, \sigma_x\,,\\
A_1^{(i)} &=\sigma_z  \,,\quad &B_1= - \cos(\mu) \, \sigma_z + \sin(\mu) \, \sigma_x \,, \label{eq:app-multi-observables2}
\end{align}
for $i \in \{1...$N-1$\}$, on the quantum state
\begin{align}\label{eq:app-multi-state}
   \ket{GHZ_N^{\theta}}&=\cos{\theta} \ket{0}^{\otimes N} + \sin{\theta} \ket{1}^{\otimes N} \,, 
\end{align}
with the relation between the parameters $\alpha$, $\theta$, and $\mu$ being the following: 
\begin{align}
\alpha = \theta\,,\quad 2 \sin^2(\mu) = \sin^2(2\theta)\,.
\end{align}
It was also proved in Ref.~\cite{multiGHZ} that the inequality \eqref{eq:app-multi-Bell} can be used to self-test the the state in Eq.~\eqref{eq:app-multi-state} for any $\theta \in (0,\pi/4]$.

\begin{rem}\label{rem:STarg-multi}
Let $\widetilde{A}_x^{(i)}$, $\widetilde{B}_y$, and $\ket{\widetilde{\psi}}$ be physical observables and a physical quantum state that achieve the value $I_\alpha^{(N) \max}$ for the Bell inequality $I_\alpha^{(N)}$. Then, these are equal to (up to local isometries) the state and observables of Eqs.~\eqref{eq:app-multi-state}, \eqref{eq:app-multi-observables}, and \eqref{eq:app-multi-observables2}.
\end{rem}

\section{Constructing multipartite EPR inequalities from Bell inequalities}\label{app:steeringin-multi}

We now define a family of multipartite EPR functionals $S_{\alpha}^{(N)}$ from the multipartite Bell inequalities $I_\alpha^{(N)} \in \I^{(N)}$. A generic $N$-partite EPR functional is defined as:
\begin{align}
S^{(N)}[\As_N] = \Tr{ \sum_{\substack{a_i\in \A, \\ x_i\in \X}} F_{a_1\ldots a_{N-1},x_1\ldots x_{N-1}} \, \sigma_{a_1\ldots a_{N-1}|x_1\ldots x_{N-1}}}\,.
\end{align} 
To transform the Bell functional to the EPR functional, notice the following relations. First, for any choice of $x_1\ldots x_{N-1}$, we can write:
\begin{align}
    \av{B}=\Tr{B \, \rho_B}=\sum_{a_1\ldots a_{N-1}} \Tr{B \, \sigma_{a_1\ldots a_{N-1}|x_1\ldots x_{N-1}}} \,.
\end{align}
Moreover, 
\begin{align}
\begin{split}
    \av{(B_0-B_1)A_1^{(i)}}&=\Tr{\id^{(1)} \otimes \ldots \otimes (\widetilde{M}_{0|1}-\widetilde{M}_{1|1})^{(i)} \otimes \ldots \otimes \id^{(N-1)} \otimes (B_0-B_1) \, \rho } \\
    &= \Tr{(B_0-B_1)  \, (\sigma_{0|1}^{(i)}-\sigma_{1|1}^{(i)})}  \,,
\end{split}
\end{align}
where $\sigma_{a|x}^{(i)}=\sum_{a_j, j \neq i}\sigma_{a_1\ldots (a_i=a) \ldots a_{N-1}|x_1 \ldots (x_i=x) \ldots x_{N-1}}$. 

Lastly, 
\begin{align}
\begin{split}
    \av{(B_0+B_1)A_0^{(1)}\ldots A_0^{(N-1)}}&=\Tr{(\widetilde{M}_{0|0}-\widetilde{M}_{1|0})^{(1)} \otimes \ldots \otimes (\widetilde{M}_{0|0}-\widetilde{M}_{1|0})^{(N-1)} \otimes (B_0+B_1) \, \rho } \\
    &= \sum_{a_1,\ldots,a_{N-1}} (-1)^{a_1+\ldots+a_{N-1}} \Tr{\widetilde{M}_{a_1|0}^{(1)} \otimes \ldots \otimes \widetilde{M}_{a_{N-1}|0}^{(N-1)} \otimes (B_0+B_1) \, \rho }  \\
    &= \sum_{a_, \ldots, a_{N-1}} (-1)^{a_1+\ldots+a_{N-1}} \Tr{ (B_0+B_1) \, \sigma_{a_1\ldots a_{N-1}|0\ldots 0} }.
\end{split}
\end{align}
Using the above relations, we define an EPR functional $S^{(N)}_{\alpha}[\As_N]$ by setting: 
\begin{align}\label{eq:theF-multi}
F_{a_1\ldots a_{N-1},0\ldots0} = (N-1) (-1)^{a_1+\ldots+a_{N-1}} (B_0 + B_1) +&  (N-1)  \frac{\cos{2\alpha}}{\sqrt{1-\cos^2{2\alpha}}} (B_0 - B_1), \\ 
F_{a_1\ldots (a_i=a) \ldots x_{N-1},x_1\ldots (x_i=1)\ldots x_{N-1}} = (-1)^{a} \frac{1}{\sqrt{1-\cos^2{2\alpha}}} &(B_0 - B_1). \nonumber
\end{align}
From the operators of Eq.~\eqref{eq:theF-multi}, we construct a family of EPR functionals as follows.
\begin{defn}\label{def:NSFT-multi} \textbf{EPR functional $S_\eta^{(N)}[\As_N]$.} \\
We define the EPR functional $S_\eta^{(N)}[\As_N]$ via the operators of Eq.~\eqref{eq:theF-multi} by taking: 
\begin{align}\label{eq:theStpar-multi}
\alpha &= \eta \,, \\ \nonumber
B_0 &= \cos(\mu) \, \sigma_z + \sin(\mu) \, \sigma_x \,,\\ \nonumber
B_1 &=- \cos(\mu) \, \sigma_z + \sin(\mu) \, \sigma_x \,, \\ \nonumber
2 \sin^2(\mu) &= \sin^2(2\eta)\,.
\end{align}
\end{defn}
The properties of the multipartite EPR functional from Definition \ref{def:NSFT-multi} can be studied in a similar manner to the  properties of the bipartite EPR functional from Definition \ref{def:SFT}. We now make a few remarks about these properties.

\begin{rem}
The maximum quantum value $S_\eta^{(N) \max}$ of $S_\eta^{(N)}[\As_N]$ is given by 
\begin{align}\label{eq:SmaxN}
S_\eta^{(N)\max} = 2\sqrt{2} \, (N-1)\,.
\end{align}
\end{rem}
\begin{proof}
The problem of optimising $S_\eta^{(N)}[\As_N]$ over quantum assemblages is equivalent to the problem optimising $I_\alpha^{(N)}$ over quantum correlations constrained on Bob's observables being those from Eq.~\eqref{eq:theStpar-multi} for $\alpha=\eta$. Since the maximum quantum violation $I_\alpha^{(N) \max}$ can indeed be achieved by the operators for Bob from Eq.~\eqref{eq:theStpar-multi}, then $S_\eta^{(N) \max} = I_\alpha^{(N) \max}$.
\end{proof}

In order to explore quantum assemblages that achieve this maximum violation, first we need to prove a claim about the interconvertibility of the quantum states of the form given in Eq.~\eqref{eq:Nass} under LOSR operations. Let us recall Corollary 8 in \cite{schmid2020standard}.

\begin{rem}[{\cite[Corollary 8]{schmid2020standard}}]\label{rem:bipartitions}
An n-partite pure state $\ket{\psi}$ can be converted to an n-partite pure state $\ket{\phi}$ by LOSR only if
\begin{align}\label{eq:bipartitions}
 \exists \ket{\zeta}, \forall \beta: 
   (\lambda_{\psi}^{(\beta)})^{\downarrow} = (\lambda_{\phi}^{(\beta)} \otimes \lambda_{\zeta}^{(\beta)})^{\downarrow}\,
\end{align}
where for a pure state $\ket{\omega}$, $\lambda_{\omega}^{(\beta)}$ denotes the vector of its squared Schmidt coefficients with respect to bipartition $\beta$ of the n-partite system.
\end{rem}
We are now in position to prove the following Lemma.

\begin{lem}\label{GHZunorder}
Let $\rho_N^{\theta}=\ket{GHZ_N^{\theta}}\bra{GHZ_N^{\theta}}$ be defined as per Eq.~\eqref{eq:Nass}. Then
\begin{align}
\rho_N^{\theta} \overset{\text{LOSR}}{\not\longrightarrow} \, \rho_N^{\phi}\,,
\end{align}
for any $\theta \neq \phi \in \left( 0,\sfrac{\pi}{4} \right]$.
\end{lem}

\begin{proof}
Let $\rho_N^{\theta}=\ket{GHZ_N^{\theta}}\bra{GHZ_N^{\theta}}$ and  $\rho_N^{\phi}=\ket{GHZ_N^{\phi}}\bra{GHZ_N^{\phi}}$ with  $\theta \neq \phi \in \left( 0,\sfrac{\pi}{4} \right]$. All bipartitions $\beta$ of $\rho_N^{\theta}$ and $\rho_N^{\phi}$ have the same Schmidt rank, which is 2 for all $\beta$. However, for all $\beta$,
the vectors of the squared Schmidt coefficients $\lambda_{\theta}^{(\beta)}$ are different than $\lambda_{\phi}^{(\beta)}$, since 

$$
\lambda_{\theta}^{(\beta)} = \begin{bmatrix} \cos^2 \theta \\ \sin^2 \theta \end{bmatrix} \quad \text{and} \quad \lambda_{\phi}^{(\beta)} = \begin{bmatrix} \cos^2 \phi \\ \sin^2 \phi \end{bmatrix}.
$$
By applying Corollary 9 of Ref.~\cite{schmid2020standard}, it follows that there exists no $\ket{\zeta}$ such that the condition specified in Eq.~\eqref{eq:bipartitions} is satisfied. From Remark \ref{rem:bipartitions} it follows that $\rho_N^{\theta}$ and  $\rho_N^{\phi}$
are such that neither converts to the other under LOSR operations.
\end{proof}

\begin{rem}\label{rem:iffN}
The multipartite quantum assemblage $\As_N^{\theta}$ satisfies $S_\eta^{(N)}[\As_N^{\theta}] = S_\eta^{(N) \max}$ if and only if $\eta=\theta$.
\end{rem}
\begin{proof}
The `if' direction is straightforward to prove, by noticing that the state and measurements used to prepare $\As_N^{\theta}$, given by Eqs.~\eqref{eq:Nass} and \eqref{eq:Nmeas}, together with Bob's observables from Eq.~\eqref{eq:theStpar-multi}, produce correlations that achieve $I_\alpha^{(N) \max}$ in the corresponding Bell experiment, where $\alpha=\eta$.

The `only if' direction follows from a proof by contradiction.
Assume that there exists $\As_N^{\theta}$ with $\theta \neq \eta$, such that $S_\eta^{(N)}[\As_N^{\theta}] = S_\eta^{(N) \max}$. 
Then, the state and measurements used to prepare $\As_N^{\theta}$, together with Bob's observables from Eq.~\eqref{eq:theStpar-multi}, become a quantum realisation of correlations that achieve $I_\alpha^{(N) \max}$ for the Bell functional $I_\alpha^{(N)}$, where $\alpha =\eta$. Then, it follows from Remark \ref{rem:STarg-multi} that there exists a local isometry in Alice's lab that takes $\rho_N^{\theta}$ to $\rho_N^{\eta}$ (both defined in Eq.~\eqref{eq:Nass}). However, such local isometry does not exist, as it would contradict Lemma \ref{GHZunorder}. Hence, it must be the case that $\theta = \eta$. 

\end{proof}

\begin{rem}\label{rem:thelastN}
Let $\As_N$ be a quantum assemblage, and let $\rho$ and $\{\{M_{a_i|x_i}\}_{a_i\in \A}\}_{x_i\in \X\,,\, i\in\{1,\ldots,N-1\}}$ be a quantum state and measurements that realise it. Then, $S_\eta^{(N)}[\As_N] = S_\eta^{(N) \max}$ if and only if: \\
- The observables $A^{\prime(i)}_0 = M_{0|0} - M_{1|0}$ and $A^{\prime(i)}_1 = M_{0|1} - M_{1|1}$ are equivalent (up to local isometries) to $A^{(i)}_0$ and $A^{(i)}_1$, respectively, from Eqs.~\eqref{eq:app-multi-observables} and \eqref{eq:app-multi-observables2} for each Alice, i.e., $i\in\{1,\ldots,N-1\}$,\\
- The quantum state $\rho$ is equivalent, up to local isometries, to $\rho_N^{\eta}$. 
\end{rem}
\begin{proof}
The `if' direction follows trivially from Remark \ref{rem:iffN}. \\
To prove the `only if' direction, notice that $\rho$ and $\{A^{\prime(i)}_j\}_{j=0,1\,,\,i\in\{1,\ldots,N-1\}}$, together with Bob's observables $\{B_j\}_{j=0,1}$ from Eq.~\eqref{eq:theStpar-multi}, provide a quantum realisation of correlations that achieve the maximum value of the Bell functional $I_\alpha^{(N)}$, with $\alpha=\eta$. From Remark \ref{rem:STarg-multi} it follows that $\rho$ and $\{A^{\prime(i)}_j\}_{j=0,1\,,\,i\in\{1,\ldots,N-1\}}$ must be equal to (up to local isometries) the state and observables of Eq.~\eqref{eq:app-multi-state},  \eqref{eq:app-multi-observables}, and \eqref{eq:app-multi-observables2}, which proves our claim. 
\end{proof}

\section{Proof of Corollary \ref{cor:multi}}\label{app:proof-cor-multi}

In this section, we give a proof of Corollary \ref{cor:multi}, which pertains to multipartite quantum assemblages.

Let us make a few remarks about the properties of the resource monotone $\M_\eta^{(N)}$ given by Definition \ref{def:MetaN}. First, note that when $\As_N$ is a quantum assemblage, its LOSR processing $\widetilde{\As}_N$ is also a quantum assemblage. Hence, $S_\eta^{(N)}[\widetilde{\As}_N]$ is upperbounded by its maximum quantum violation, which is given by $S_\eta^{(N)}[\As^{\eta}_N]$:
\begin{align}
\M_\eta^{(N)}[\As_N] \leq S_\eta^{(N)}[\As^{\eta}_N] \quad \forall \, \eta \in \left(0,\sfrac{\pi}{4}\right]\,.
\end{align}
Second, note that $\M_\eta^{(N)}[\As^{\eta}_N] = S_\eta^{(N)}[\As^{\eta}_N]$. These two properties set the stage for the following theorem:

\begin{thm}\label{thm:4-multi}
For a monotone $\M_\eta^{(N)}$ given in definition \ref{def:MetaN} and an EPR functional $S_\eta^{(N)}$ specified in definition \ref{def:NSFT-multi}, if $\theta \neq \eta$, then $\M_\eta^{(N)}[\As^{\theta}_N] < S_\eta^{(N)}[\As^{\eta}_N]$.
\end{thm}
\begin{proof}
Let us prove this by contradiction. Our starting assumption is that there exists a pair $(\theta\,,\, \eta)$ with $\theta \neq \eta$, such that $\M_\eta^{(N)}[\As^{\theta}_N] = S_\eta^{(N)}[\As^{\eta}_N]$. Then, one of the two should happen: 

\bigskip

\noindent \underline{First case:}  $\M_\eta^{(N)}[\As^{\theta}_N] = S_\eta^{(N)}[\As^{\theta}_N]$. \\
In this case, the solution to the maximisation problem in the computation of $\M_\eta^{(N)}$ is achieved  by $\As^{\theta}_N$ itself. \\
Our starting assumption then tells us that $S_\eta^{(N)}[\As^{\theta}_N] = S_\eta^{(N)}[\As^{\eta}_N]=S_\eta^{(N) \max}$.\\
From Remark \ref{rem:iffN} it follows that necessarily $\theta=\eta$, which contradicts our initial condition. 

\bigskip

\noindent \underline{Second case:}  $\M_\eta^{(N)}[\As^{\theta}_N] = S_\eta^{(N)}[\widetilde{\As}_N]$, with $\As^{\theta}_N \overset{\mathrm{LOSR}}{\longrightarrow} \widetilde{\As}_N$ . \\
In this case, the solution to the maximisation problem in the computation of $\M_\eta$ is achieved by an LOSR processing of $\As^{\theta}_N$. 
Our starting assumption then tells us that $S_\eta^{(N)}[\widetilde{\As}_N] = S_\eta^{(N)}[\As^{\eta}_N]=S_\eta^{(N) \max}$.
Let $\rho$ and $\{\{M_{a_i|x_i}\}_{a_i\in A}\}_{x_i\in X\,,\,i\in\{1,\ldots,N-1\}}$ be any quantum state and measurements that realise the quantum assemblage $\widetilde{\As}_N$ in $N$-partite EPR scenario. From Remark \ref{rem:thelastN}, we know that $\rho$ is equivalent to $\rho^{\eta}_N$ up to local isometries. But since local isometries are free LOSR operations, this means that $\rho^{\theta}_N\overset{\mathrm{LOSR}}{\longrightarrow} \rho^{\eta}_N$, which contradicts Lemma \ref{GHZunorder}.
\end{proof}

We can now prove that $\FA^{(N)}$ is composed of pairwise unordered resources. In analogue to the bipartite case, to prove this claim it suffices to find a set of EPR LOSR monotones $\FM = \{\M_j\}$ such that, for every pair $(\theta_1,\theta_2)$ there exists a pair $(\M_1,\M_2)$ with $\M_1(\As^{\theta_1}_N) > \M_1(\As^{\theta_2}_N) \,,\quad \M_2(\As^{\theta_1}_N) < \M_2(\As^{\theta_2}_N).$ As we see next, this is achieved by choosing $\M_k = \M^{(N)}_{\theta_k}$.

\begin{thm}\label{th:app-multi}
For every pair $\theta_1 \neq \theta_2$, the monotones $M_{\theta_1}^{(N)}$ and $M_{\theta_2}^{(N)}$ given by Definition \ref{def:MetaN} satisfy 
\begin{align} \label{eq:6N}
\M_{\theta_1}^{(N)}(\As^{\theta_1}_N) &> \M_{\theta_1}^{(N)}(\As^{\theta_2}_N) \,, \\
\M_{\theta_2}^{(N)}(\As^{\theta_1}_N) &< \M_{\theta_2}^{(N)}(\As^{\theta_2}_N) \,. \label{eq:7N}
\end{align}
\end{thm}
\begin{proof}
Let us first prove Eq.~\eqref{eq:6N}. We know that $\M_{\theta_1}^{(N)}(\As^{\theta_1}_N) = S_{\theta_1}^{(N)}[\As^{\theta_1}_N]$. However, since $\theta_1 \neq \theta_2$, Theorem \ref{thm:4-multi} implies that $\M_{\theta_1}^{(N)}(\As^{\theta_2}_N) < S_{\theta_1}^{(N)}[\As^{\theta_1}_N]$. Therefore,  Eq.~\eqref{eq:6N} follows. \\
The proof of  Eq.~\eqref{eq:7N} follows similarly. 
\end{proof}

Corollary \ref{cor:multi} follows directly from Theorem \ref{th:app-multi}.

\section{Proof of Proposition~\ref{prop:pre-order}}\label{app:proof-prop-preorder}

Let us focus on the family of assemblages $\As^{\theta,p}_{\A|\X}$ defined by Eq.~\eqref{eq:thefam}. In Corollary~\ref{cor:bip} we showed that under LOSR operations, all assemblages $\As^{\theta,1}_{\A|\X}$ are incomparable for $\theta \in \left(0,\sfrac{\pi}{4}\right]$. We will now show that under S-1W-LOCC, the assemblage $\As^{\pi/4,1}_{\A|\X}$ is above all others in the family; hence, the pre-orders in the two different resource theories are different. 

The S-1W-LOCC map that maps $\As^{\pi/4,1}_{\A|\X}$ to $\As^{\theta,1}_{\A|\X}$ is very simple: the instrument has a single outcome, say $0$, and the corresponding map $\mathcal{E}_{0}(\cdot)$ on Bob's system is $M_{0}\cdot M^{\dagger}_{0}$, where $M_{0}=\cos(\theta)\ket{0}\bra{0}+\sin(\theta)\ket{1}\bra{1}$ for the value of $\theta$ corresponding to the assemblage $\As^{\theta,1}_{\A|\X}$. A straightforward calculation confirms that this map converts from $\As^{\pi/4,1}_{\A|\X}$ to the target assemblage, and that the outcome $0$ occurs with probability $1/2$. 
Remarkably, this map does not even require communication and is actually an example of a stochastic-LOSR operation, i.e., an LOSR map without the requirement that the map applied on Bob's subsystem is trace-preserving, but occurs with non-zero probability.

One could argue that the reason the pre-orders are different for stochastic-1W-LOCC and LOSR is because of the stochastic nature of the former. 
However, the pre-orders for deterministic 1W-LOCC maps and LOSR is again different; as before, there is such a deterministic 1W-LOCC map that converts $\As^{\pi/4,1}_{\A|\X}$ to $\As^{\theta,1}_{\A|\X}$. It is defined as follows: there are two outcomes $\omega\in\{0,1\}$ where, for outcome $0$, $\mathcal{E}_{0}(\cdot)=M_{0}\cdot M^{\dagger}_{0}$ (as before), and for outcome $1$, $\mathcal{E}_{1}=M_{1}\cdot M^{\dagger}_{1}$, where $M_1=\sin(\theta)\ket{1}\bra{0}+\cos(\theta)\ket{0}\bra{1}$. One can readily confirm that this is trace-preserving, since $M^{\dagger}_{0}M_{0}+M^{\dagger}_{1}M_{1}=\mathbb{I}$. Bob then communicates the outcome $\omega$ to Alice, whose input $x$ is equal to the input $x'$ to her measurement device.
Once she gets the output $a'$ of the measurement, the final output becomes $a=a'\oplus x\omega \oplus \omega$. To summarise, Alice will flip the output of her measurement if and only if $x=0$ and $\omega=1$. 
\end{document}

%% file: Figs/SDP-results.tex
\usetikzlibrary{arrows.meta}
\colorlet{lgray}{gray!50}
\tikzset{>={Latex[width=3mm,length=3mm]}}

\begin{tikzpicture}
\foreach \x in {0,1,2}{
  \foreach \y in {1,2}{
      \node[circle,fill=black,draw] (\x;\y)  at (2*\x,2*\y)  {} ;
    }
} 

\foreach \x in {1,2}{
  \foreach \y in {0}{
      \node[circle,fill=black,draw] (\x;\y)  at (2*\x,2*\y)  {} ;
    }
}

\node[circle,draw=lgray,fill=lgray,draw] (0;0) at (0,0) {} ;

%trivial conversions
\foreach \i in {0,1,2}{
      {\draw[dashed, ->,draw=lgray] (\i;1) -- (\i;0);}
      {\draw[dashed,->,draw=lgray] (\i;2) -- (\i;1);}
    }
    
%pi/4
%\foreach \i in {0,1}{
  %    {\draw[->,draw=black] (2;2) -- (\i;0);}

 %   }
\draw[->,draw=black] (2;2) -- (0;1); 
\draw[->,draw=black] (2;2) -- (1;0);

%pi/6
%\foreach \i in {0,2}{
     % {\draw[->,draw=black] (1;2) -- (\i;0);}
      %{\draw[->,draw=black] (1;2) -- (\i;1);}
   % }
\draw[->,draw=black] (1;2) -- (2;1);

%pi/12
\draw[->,draw=black] (0;2) -- (1;0); 

%coming out of p=0.9:
%pi/4
\draw[->,draw=black] (2;1) -- (1;0);

%pi/6
%\foreach \i in {0,2}{
%      {\draw[->,draw=black] (1;1) -- (\i;0);}
%    }
\draw[->,draw=black] (1;1) -- (0;1); 
{\draw[->,draw=black] (1;1) -- (2;0);}

%coming out of p=0.8:
%pi/4
\draw[dashed,->,draw=lgray] (2;0) to [out=-150,in=-30] (0;0); 

%pi/6
\draw[dashed,->,draw=lgray] (1;0) -- (0;0); 

%axis
\node[] at (2*3,2*1) {\Large $p$};
\node[] at (2*2.5,2*0) {\Large 0.8};
\node[] at (2*2.5,2*1) {\Large 0.9};
\node[] at (2*2.5,2*2) {\Large 1.0};

\node[] at (2*1,2*3) {\LARGE $\theta$};
\node[] at (2*0,2*2.5) {\LARGE $\frac{\pi}{12}$};
\node[] at (2*1,2*2.5) {\LARGE $\frac{\pi}{6}$};
\node[] at (2*2,2*2.5) {\LARGE $\frac{\pi}{4}$};

\end{tikzpicture}